\documentclass[11pt,a4paper]{article}
\usepackage{hyperref}
\usepackage{subcaption}
\usepackage{fullpage}

\newcommand\eop\qed

\newcommand{\ifConferenceVersion}{\iffalse}
\newcommand{\ifJournalVersion}{\iftrue}

\newcommand{\mathx}{\ifConferenceVersion $ \else $$\fi}

\usepackage{verbatim}

\providecommand{\InJournal}{}
\providecommand{\InConference}{}

\ifJournalVersion
    \renewcommand{\InJournal}{}
    \renewcommand{\InConference}[1]{}
\else
    \renewcommand{\InJournal}[1]{}
    \renewcommand{\InConference}{}
\fi

\newcommand{\DD}[1]{\textcolor{blue}{\textbf{DD: #1}}}
\newcommand{\AK}[1]{\textcolor{green}{\textbf{AK: #1}}}
\newcommand{\MC}[1]{\textcolor{magenta}{\textbf{MC: #1}}}
\newcommand{\PU}[1]{\textcolor{red}{\textbf{PU: #1}}}

\renewcommand{\DD}[1]{}
\renewcommand{\AK}[1]{}
\renewcommand{\MC}[1]{}
\renewcommand{\PU}[1]{}

\usepackage[T1]{fontenc}
\usepackage{latexsym}
\usepackage{color,graphicx}
\usepackage{amssymb}
\usepackage{amsmath}
\usepackage{amsthm}
\usepackage{enumitem}
\usepackage{authblk}
\usepackage{algorithm}
\usepackage[utf8]{inputenc}

\makeatletter

\@addtoreset{algorithm}{section}
\makeatother

\numberwithin{equation}{section}

\usepackage[noend]{algpseudocode}
\usepackage{multirow}
\usepackage{arydshln}
\usepackage{booktabs}

\newcommand{\comega}{(2c+1)\omega}

\newenvironment{rcases}
  {\left.\begin{aligned}}
  {\end{aligned}\right\rbrace}

\newcommand{\stX}{\hspace{0.1cm}\bigl|\bigr.\hspace{0.1cm}}
\newcommand{\card}[1]{\left|#1\right|}
\newcommand{\roundup}[2]{\lceil #1 \rceil_{#2}}

\newcommand{\tinserted}[1]{t_{end}^{#1}}

\theoremstyle{plain}

\newtheorem{theorem}{Theorem}[section]
\newtheorem{lemma}[theorem]{Lemma}

\newtheorem{corollary}[theorem]{Corollary}

\newtheorem{proposition}[theorem]{Proposition}
\newtheorem{observation}[theorem]{Observation}

\renewcommand{\root}[1]{r(#1)}
\newcommand{\opt}[1]{\textup{\texttt{OPT}}(#1)}
\newcommand{\cost}{\textup{\texttt{COST}}}

\newcommand{\Q}{\textup{\texttt{Q}}}

\newcommand{\w}{w}
\newcommand{\cA}{\mathcal{A}}
\newcommand{\cO}{\mathcal{O}}
\newcommand{\cR}{\mathcal{R}}
\newcommand{\N}{\mathbb{N}}

\newcommand{\queriedSubtree}[3]{#1\langle #2,#3\rangle}

\newcommand{\eps}{\varepsilon}
\newcommand{\s}{\hat S}
\newcommand{\costo}{\cost^{( \omega,c)}}
\newcommand{\labell}{\ell}
\newcommand{\longchains}[1]{\xi(#1)}

\makeatletter
\newcommand{\StatexIndent}[1][3]{%
  \setlength\@tempdima{\algorithmicindent}%
  \Statex\hskip\dimexpr#1\@tempdima\relax}
\makeatother

\usepackage{chngcntr}
\counterwithin{table}{section}
\counterwithin{figure}{section}
\counterwithin{theorem}{section}

\title{%
\textbf{Approximation Strategies for Generalized Binary Search in Weighted Trees}}
\author[1]{Dariusz Dereniowski}
\affil[1]{Faculty of Electronics, Telecommunications and Informatics,
\mbox{Gda\'nsk University of Technology, Poland}}
\author[2]{Adrian Kosowski}
\affil[2]{Inria Paris and IRIF, Universit\'e Paris Diderot, France}
\author[3]{Przemys\l{}aw Uzna\'nski}
\affil[3]{Department of Computer Science, ETH Z\"{u}rich, Switzerland}
\author[2]{Mengchuan Zou}
\date{}

\InJournal{
}

\begin{document}

\maketitle

\begin{abstract}
We consider the following generalization of the binary search problem. A search strategy is required to locate an unknown target node $t$ in a given tree $T$. Upon querying a node $v$ of the tree, the strategy receives as a reply an indication of the connected component of $T\setminus\{v\}$ containing the target $t$. The cost of querying each node is given by a known non-negative weight function, and the considered objective is to minimize the total query cost for a worst-case choice of the target.

Designing an optimal strategy for a weighted tree search instance is known to be strongly NP-hard, in contrast to the unweighted variant of the problem which can be solved optimally in linear time. Here, we show that weighted tree search admits a quasi-polynomial time approximation scheme (QPTAS): for any $0 < \varepsilon < 1$, there exists a $(1+\varepsilon)$-approximation strategy with a computation time of $n^{O(\log n / \eps^2)}$. Thus, the problem is not APX-hard, unless $NP \subseteq DTIME(n^{O(\log n)})$. By applying a generic reduction, we obtain as a corollary that the studied problem admits a polynomial-time $O(\sqrt{\log n})$-approximation. 
This improves previous $\hat O(\log n)$-approximation approaches, where the $\hat O$-notation disregards $O(\mathrm{poly}\log\log n)$-factors.
\end{abstract}



\section{Introduction}

In this work we consider a generalization of the fundamental problem of searching for an element in a sorted array.
This problem can be seen, using graph-theoretic terms, as a problem of searching for a target node in a path, where each query reveals on which `side' of the queried node the target node lies.
The generalization we study is two-fold: a more general structure of a tree is considered and we assume non-uniform query times.
Thus, our problem can be stated as follows.
Given a node-weighted input tree $T$ (in which the query time of a node is provided as its weight), design a search strategy (sometimes called a decision tree) that locates a hidden \emph{target node} $x$ by asking \emph{queries}.
Each query selects a node $v$ in $T$ and after the time that equals the weight of the selected node, a reply is given:
the reply is either `yes' which implies that $v$ is the target node and thus the search terminates, or it is `no' in which case the search strategy receives the edge outgoing from $v$ that belongs to the shortest path between $u$ and $v$.
The goal is to design a search strategy that locates the target node and minimizes the search time in the worst case.

The vertex search problem is more general than its `edge variant' that has been more extensively studied.
In the latter problem one selects an edge $e$ of an edge-weighted tree $T=(V,E,w)$ in a query and learns in which of the two components of $T-e$ the target node is located.
Indeed, this edge variant can be reduced to our problem as follows: first assign a `large' weight to each node of $T$ (for example, one plus the sum of the weights of all edges in the graph) and then subdivide each edge $e$ of $T$ giving to the new node the weight of the original edge, $w(e)$. It is apparent that an optimal search strategy for the new node-weighted tree should never query the nodes with large weights, thus immediately providing a search strategy for the edge variant of $T$.

We also point out that the considered problem, as well as the edge variant, being quite fundamental, were historically introduced several times under different names: minimum height elimination trees \cite{Pothen88}, ordered colourings \cite{KatchalskiMS95}, node and edge rankings \cite{IyerRV88}, tree-depth \cite{NesetrilM06} or LIFO-search \cite{GiannopoulouHT12}.

Table~\ref{tab:res} summarizes the complexity status of the node-query model (in case of unweighted paths in both cases the solution is the classical binary search algorithm) and places our result in the general context.


\setlength\dashlinedash{0.2pt}
\setlength\dashlinegap{1.5pt}
\setlength\arrayrulewidth{0.3pt}
\bgroup\def\arraystretch{1.5}

\begin{table}[htb]
\caption{Computational complexity of the search problem in different graph classes, including our results for weighted trees. Completeness results refer to the decision version of the problem.}
\label{tab:res}

{
\small
\centering
\begin{tabular}{rcc}
\toprule
\emph{Graph class}                 & \emph{Unweighted}                           & \emph{Weighted}                  \\
\midrule
Paths:           & exact in $O(n)$ time                               & exact in $O(n^2)$ time \cite{CicaleseJLV12} \\
\cdashline{2-3}

\multirow{3}{*}{Trees:}           & \multirow{3}{*}{exact in $O(n)$ time \cite{OnakP06,Schaffer89}}     & strongly NP-complete \cite{DereniowskiN06} \\
                &                                       &$(1+\eps)$-approx. in $n^{O(\log n / \eps)}$ time  (Thm.~\ref{thm:qptas}) \\
                &                                       & $O(\sqrt{\log n})$-approx. in poly-time (Thm.~\ref{thm:recursive-algo})\\
\cdashline{2-3}

\multirow{2}{*}{Undirected:}   & exact in $n^{O(\log n)}$ time \cite{Emamjomeh-ZadehKS16} & PSPACE-complete \cite{Emamjomeh-ZadehKS16} \\
& $O(\log n)$-approx. in poly-time \cite{Emamjomeh-ZadehKS16} & $O(\log n)$-approx. in poly-time \cite{Emamjomeh-ZadehKS16}\\

\cdashline{2-3}
Directed: & PSPACE-complete \cite{Emamjomeh-ZadehKS16}                     & PSPACE-complete \cite{Emamjomeh-ZadehKS16}\\
\bottomrule
\end{tabular}
}
\end{table}

\subsection {State-of-the-Art}


In this work we focus on the worst case search time for a given input graph and we only remark that other optimization criteria has been also considered \cite{CicaleseJLM11,LaberMP02,LaberM11,SzwarcfiterNBOCZ03}.
For other closely related models and corresponding results see e.g. \cite{ArkinMMRS98,HeeringaIT11,LaberN04,LinialS85,Steiner87}.

\subparagraph{The node-query model.}
An optimal search strategy can be computed in linear-time for an unweighted tree \cite{OnakP06,Schaffer89}.
The number of queries performed in the worst case may vary from being constant (for a star one query is enough) to being at most $\log_2 n$ for any tree \cite{OnakP06} (by always querying a node that halves the search space).
Several following results have been obtained in \cite{Emamjomeh-ZadehKS16}.
First, it turns out that $\log_2 n$ queries are always sufficient for general simple graphs and this implies a $O(m^{\log_2 n}n^2\log n)$-time optimal algorithm for arbitrary unweighted graphs. The algorithm which performs $\log_2 n$ queries also serves as a $O(\log n)$-approximation algorithm, also for the weighted version of the problem.
(We remark that in the weighted case, the algorithms in~\cite{Emamjomeh-ZadehKS16} sometimes have an approximation ratio of $\Theta(\log n)$, even in the tree scenario we study in this work.)
On the other hand, it is shown in the same work that an optimal algorithm (for unweighted case) with a running time of $O(n^{o(\log n)})$ would be in contradiction with the Exponential-Time-Hypothesis\InJournal{, and for $\varepsilon>0$, $O(m^{(1-\varepsilon)\log n})$ would be in contradiction with the Strong Exponential-Time-Hypothesis}.
When weighted graphs are considered, the problem becomes PSPACE-complete.
\InJournal{Also, a generalization to directed graphs also turns out to be PSPACE-complete.}

\InJournal{
We also refer the interested reader to further works that consider a probabilistic version of the problem, where the answer to a query is correct with some probability $p>\frac{1}{2}$ \cite{Ben-OrH08,Emamjomeh-ZadehKS16,FeigeRPU94,KarpK07}.
In particular, for any $p>\frac{1}{2}$ and any undirected unweighted graph, a search strategy can be computed that finds the target node with probability $1-\delta$ using $(1-\delta)\frac{\log_2n}{1-H(p)}+o(\log n)+O(\log^2\frac{1}{\delta})$ queries in expectation, where $H(p)=-p\log_2p-(1-p)\log_2(1-p)$ is the entropy function.
See \cite{RivestMKWS80} for a model in which a fixed number of queries can be answered incorrectly during a binary search.
}

\subparagraph{The edge-query model.}
In the case of unweighted trees, an optimal search strategy can be computed in linear time \cite{LamY01,MozesOW08}.
(See~\cite{Dereniowski08} for a correspondence between edge rankings and the searching problem.)
The problem of computational complexity for weighted trees attracted a lot of attention.
On the negative side, it has been proved that it is strongly NP-hard to compute an optimal search strategy \cite{Dereniowski06} for bounded diameter trees, which has been improved by showing hardness for several specific topologies: trees of diameter at most 6, trees of degree at most 3 \cite{CicaleseJLV12} and spiders \cite{CicaleseKLPV14} (trees having at most one node of degree greater than two).
On the other hand, polynomial-time algorithms exist for weighted trees of diameter at most 5 and weighted paths \cite{CicaleseJLV12}.
We note that for weighted paths there exists a linear-time but approximate solution given in \cite{LaberMP02}.
For approximate polynomial-time solutions, a simple $O(\log n)$-approximation has been given in \cite{Dereniowski06} and a $O(\log n/\log \log \log n)$-approximate solution is given in \cite{CicaleseJLV12}.
Then, the best known approximation ratio has been further improved to $O(\log n/\log\log n)$ in \cite{CicaleseKLPV14}.

Some bounds on the number of queries for unweighted trees have been developed.
Observe that an optimal search strategy needs to perform at least $\log_2 n$ queries in the worst case.
However, there exist trees of maximum degree $\Delta$ that require $\Delta\log_{\Delta+1}n$ queries \cite{Ben-AsherF97}.
On the other hand, $\Theta(\Delta\log n)$ queries are always sufficient for each tree \cite{Ben-AsherF97}, which has been improved to $(\Delta+1)\log_{\Delta}n$ \cite{LaberN01}, $\Delta\log_{\Delta} n$ \cite{DereniowskiK06} and $1+\frac{\Delta-1}{\log_2(\Delta+1)-1}\log_2 n$ \cite{Emamjomeh-ZadehKS16}.

\InJournal{\subparagraph{Searching partial orders.}
The problem of searching a partial order with uniform query times is NP-complete even for partial orders with maximum element and bounded height Hasse diagram \cite{CarmoDKL04,Dereniowski08}.
For some algorithmic solutions for random partial orders see \cite{CarmoDKL04}.
For a given partial order $P$ with maximum element, an optimal solution can be obtained by computing a branching $B$ (a directed spanning tree with one target) of the directed graph representing $P$ and then finding a search strategy for the branching, as any search strategy for $B$ also provides a feasible search for $P$ \cite{Dereniowski08}.
Since computing an optimal search strategy for $B$ can be done efficiently (through the equivalence to the edge-query model), finding the right branching is a challenge.
This approach has been used in \cite{Dereniowski08} to obtain an $O(\log n/\log\log n)$-approximation polynomial time algorithm for partial orders with a maximum element.

We remark that searching a partial order with a maximum element or with a minimum element are essentially quite different.
For the latter case a linear-time algorithm with additive error of 1 has been given in \cite{OnakP06}.
As observed in \cite{Dereniowski08}, the problem of searching in tree-like partial orders with a maximum element (which corresponds to the edge-query model in trees) is equivalent to the edge ranking problem.}

\subsection{Organization of the Paper}
The aim of Section~\ref{sec:preliminaries} is to give the necessary notation and a formal statement of the problem (Sections~\ref{sec:querymodel} and~\ref{sec:strategydef}) and to provide two different but equivalent problem formulations that will be more convenient for our analysis.
As opposed to the classical problem formulation in which a strategy is seen as a \emph{decision tree}, Section~\ref{sec:querysequences} restates the problem in such a way that with each vertex $v$ of the input tree we associate a sequence of vertices that need to be iteratively queried when $v$ is the root of the current subtree that contains the target node.
In Section~\ref{sec:schedules} we extend this approach by associating with each vertex a sequence of not only vertices to be queried but also time points of the queries.

The latter problem formulation is suitable for a dynamic programming algorithm provided in Section~\ref{sec:qptas}.
In this section we introduce an auxiliary, slightly modified measure of the cost of a search strategy.
First we provide a quasi-polynomial time dynamic programming scheme that provides an arbitrarily good approximation of the output search strategy with respect to this modified cost (the analysis is deferred to Section~\ref{sec:dp}), and then we prove that the new measure is sufficiently close to the original one (the analysis is deferred to Section~\ref{sec:blackbox}).
These two facts provide the quasi-polynomial time scheme for the tree search problem, achieving  a $(1+\varepsilon)$-approximation with a computation time of $n^{O(\log n / \eps^2)}$, for any $0 < \varepsilon < 1$.

In Section~\ref{sec:approx} we observe how to use the above algorithm to derive a polynomial-time $O(\sqrt{\log n})$-approximation algorithm for the tree search problem. This is done by a divide and conquer approach: a sufficiently small subtree $T^*$ of the input tree $T$ is first computed so that the quasi-polynomial time algorithm runs in polynomial (in the size of $T$) time for $T^*$. This decomposes the problem: having a search strategy for $T^*$, the search strategies for $T-T^*$ are computed recursively. \InJournal{Details of the approach are provided in Section~\ref{sec:algosqrt}.}

\section{Preliminaries} \label{sec:preliminaries}

\subsection{Notation and Query Model} \label{sec:querymodel}

We now recall the problem of searching of an unknown target node $x$ by performing queries on the vertices of a given node-weighted rooted tree $T=(V,E,\w)$ with weight function $\w\colon V\to\mathbb{R}_+$.
Each \emph{query} selects one vertex $v$ of $T$ and after $\w(v)$ time units receives an answer: either the query returns \emph{true}, meaning that $x=v$, or it returns a neighbor $u$ of $v$ which lies closer to the target $x$ than $v$. Since we assume that the queried graph $T$ is a tree, such a neighbor $u$ is unique and is equivalently described as the unique neighbor of $v$ belonging to the same connected component of $T\setminus \{v\}$ as $x$.

All trees we consider are rooted. Given a tree $T$, the root is denoted by $\root{T}$.
For a node $v\in V$, we denote by $T_v$ the subtree of $T$ rooted at $v$. For any subset $V' \subseteq V$ (respectively, $E'\subseteq E$) we denote by $T[V']$ (resp., $T[E']$) the minimal subtree of $T$ containing all nodes from $V'$ (resp., all edges from $E'$). For $v\in V$, $N(v)$ is the set of neighbors of $v$ in $T$.

For $U\subseteq V$ and a target node $x\notin U$, there exists a unique maximal subtree of $T \setminus U$ that contains $x$; we will denote this subtree by $\queriedSubtree{T}{U}{x}$.


We denote $|V|=n$.
We will assume w.l.o.g.\ that the maximum weight of a vertex is normalized to $1$. (This normalization is immediately obtained by a proportional scaling of all units of cost.) We will also assume w.l.o.g.\ that the weight function satisfies the following \emph{star condition}:
\mathx
\text{for all $v\in V$, } w(v) \leq \sum_{u \in N(v)} w(u).
\mathx
Observe that if this condition is not fulfilled, i.e., for some vertex $v$ will have $w(v) > \sum_{u \in N(v)} w(u)$, then vertex $v$ will never be queried by any optimal strategy in $v$, since a query to $v$ can then be replaced by a sequence of queries to all neighbors of $v$, obtaining not less information at strictly smaller cost. In general, given an instance which does not satisfy the star condition, we enforce it by performing all necessary weight replacements $w(v) \gets \min \{w(v), \sum_{u \in N(v)} w(u)\}$, for $v \in V$.

For $a, \omega \in \mathbb R_{\geq 0}$, we denote the rounding of $a$ down (up) to the nearest multiple of $\omega$ as $\lfloor a \rfloor_{\omega} = \omega\lfloor a/\omega\rfloor$ and $\lceil a \rceil_{\omega} = \omega\lceil a/\omega\rceil$, respectively.

\subsection{Definition of a Search Strategy} \label{sec:strategydef}

\emph{A search strategy} $\cA$ for a rooted tree $T=(V,E,\w)$ is an adaptive algorithm which defines successive queries to the tree, based on responses to previous queries, with the objective of locating the target vertex in a finite number of steps.
Note that search strategies can be seen as decision trees in which each node represents a subset of vertices of $T$ that contains $x$, with leaves representing singletons consisting of $x$.

Let $\Q_\cA(T,x)$ be the time-ordering (sequence) of queries performed by strategy $\cA$ on tree $T$ to find a target vertex $x$, with $\Q_{\cA,i}(T,x)$ denoting the $i$-th queried vertex in this time ordering, $1\leq i \leq |\Q_\cA(T,x)|$.

We denote by
\mathx
\cost_\cA(T,x) = \sum_{i=1}^{|\Q_\cA(T,x)|} w(\Q_{\cA,i}(T,x))
\mathx
the sum of weights of all vertices queried by $\cA$ with $x$ being the target node, i.e., the time after which $\cA$ finishes.
Let
\mathx
\cost_\cA(T)=\max_{x\in V}\cost_\cA(T,x)
\mathx
be the \emph{cost of $\cA$}.
We define the \emph{cost of $T$} to be
\mathx
\opt{T}=\min\{\cost_\cA(T)\stX \cA\textup{ is a search strategy for }T\}.
\mathx
We say that a search strategy is \emph{optimal} for $T$ if its cost equals $\opt{T}$.

As a consequence of normalization and the star condition, we have the following bound.

\begin{observation} \label{obs:opt-bounds}
For any tree $T$, we have $1\leq\opt{T}\leq \lceil \log_2 n \rceil$.
\end{observation}
\InJournal{
\begin{proof}
By the star condition, considering any vertex $v \in V$ as the target, we trivially have
\[\opt T \geq \inf_\cA \cost_{\cA} (T,v) \geq \inf_\cA \cost_{\cA}  (T [\{v\} \cup N(v)],v) \geq w(v).\]
Thus, $\opt T \geq \max_{v\in V} w(v) = 1$, which gives the first inequality.

For the second inequality, we observe that applying to tree $T$ the optimal search strategy for unweighted trees, we can locate the target in at most $\lceil \log_2 n \rceil$ queries (cf.~e.g.~\cite{KatchalskiMS95,OnakP06}). Since the cost of each query is at most $1$, the claim follows.
\end{proof}
}
\InConference{All omitted proofs are provided in the Appendix.}


We also introduce the following notation. If the first $\card{U}$ queried vertices by a search strategy $\cA$ are exactly the vertices in $U$, $U = \{\Q_{\cA,i}(T,x) : 1\leq i \leq |U|\}$, then we say that $\cA$ \emph{reaches $\queriedSubtree{T}{U}{x}$ through $U$}, and $\w(U)$ is the \emph{cost of reaching $\queriedSubtree{T}{U}{x}$ by $\cA$}. We also say that we receive an `up' reply to a query to a vertex $v$ if the root of the tree remaining to be searched remains unchanged by the query, i.e., $r(\queriedSubtree{T}{U}{x}) = r(\queriedSubtree{T}{U \cup \{v\}}{x})$, and we call the reply a `down' reply when the root of the remaining tree changes, i.e., $r(\queriedSubtree{T}{U}{x}) \neq r(\queriedSubtree{T}{U \cup \{v\}}{x})$. 
Without loss of generality, after having performed a sequence of queries $U$, we can assume that the tree $\queriedSubtree{T}{U}{x}$ is known to the strategy.

\subsection{Query Sequences and Stable Strategies} \label{sec:querysequences}

By a slight abuse of notation, we will call a search strategy \emph{polynomial-time} if it can be implemented using a dynamic (adaptive) algorithm which computes the next queried vertex in polynomial time.

We give most of our attention herein to search strategies in trees which admit a natural (non-adaptive, polynomial-space) representation called a \emph{query sequence assignment}. Formally, for a rooted tree $T$, the \emph{query sequence assignment} $S$ is a function $S : V \to V^*$, which assigns to each vertex $v\in V$ an ordered sequence of vertices $S(v)$, known as the \emph{query sequence} of $v$.
The query sequence assignment directly induces a strategy $\cA_{S}$, presented as Algorithm~\ref{alg:As}. Intuitively, the strategy processes successive queries from the sequence $S(v)$, where $v$ is the root vertex of the current search tree, $v = r(\queriedSubtree{T}{U}{x})$, where $U$ is the set of queries performed so far.
This processing is performed in such a way that the strategy iteratively takes the first vertex in $S(v)$ that belongs to $\queriedSubtree{T}{U}{x}$ and queries it.
As soon as the root of the search tree changes, the procedure starts processing queries from the sequence of the new root, which belong to the remaining search tree. The procedure terminates as soon as $\queriedSubtree{T}{U}{x}$ has been reduced to a single vertex, which is necessarily the target $x$. 

\begin{algorithm}
\small
\caption{Search strategy $\cA_S$ for a query sequence assignment $S$}
\label{alg:As}
\begin{algorithmic}[1]

\State $v \gets r(T)$ \quad // stores current root
\State $U \gets \emptyset$
\While{$|\queriedSubtree{T}{U}{x}|>1$}
\For {$u\in S(v)$}
\If {$u\in \queriedSubtree{T}{U}{x}$} \quad // $u$ is the first vertex in $S(v)$ that belongs to $\queriedSubtree{T}{U}{x}$
\State \Call{QueryVertex}{$u$}
\State $U \gets U\cup \{u\}$
\If{$v\neq r(\queriedSubtree{T}{U}{x})$} \quad // query reply is `down'
\State $v\gets r(\queriedSubtree{T}{U}{x})$\label{alg:As:linerootchange}
\State \textbf{break} \quad // for loop
\EndIf
\EndIf
\EndFor
\EndWhile

\end{algorithmic}
\end{algorithm}

In what follows, in order to show that our approximation strategies are polynomial-time, we will confine ourselves to presenting a polynomial-time algorithm which outputs an appropriate sequence assignment.


A sequence assignment is called \emph{stable} if the replacement of line~\ref{alg:As:linerootchange} in Algorithm~\ref{alg:As} by any assignment of the form $v \gets v''$, where $v''$ is an arbitrary vertex which is promised to lie on the path from $\root{\queriedSubtree{T}{U}{x}}$ to the target $x$, always results in a strategy which performs a (not necessarily strict) subsequence of the sequence of queries performed by the original strategy $\cA_S$. Sequence assignments computed on trees with a bottom-up approach usually have the stability property; we provide a proof of stability for one of our main routines in Section~\ref{sec:dp}.

Without loss of generality, we will also assume that if $v \in S(v)$, then $v$ is the last element of $S(v)$. Indeed, when considering a subtree rooted at $v$, after a query to $v$, if $v$ was not the target, then the root of the considered subtree will change to one of the children of $v$, hence any subsequent elements of $S(v)$ may be removed without changing the strategy.

\subsection{Strategies Based on Consistent Schedules} \label{sec:schedules}

Intuitively, we may represent search strategies by a schedule consisting of some number of jobs, with each job being associated to querying a node in the tree  (cf.\ e.g.\ \cite{IyerRV88b,Liu86,Liu90}). Each job has a fixed processing time, which is set to the weight of a node. Formally, in this work we will refer to the schedule $\s$ only in the very precise context of search strategies $\cA_S$ based on some query sequence assignment $S$. The \emph{schedule assignment} $\s$ is the following extension of the sequence assignment $S$, which additionally encodes the starting time of search query job. If the query sequence $S$ of a node $v$ is of the form $S(v) = (v_1, \ldots, v_k)$, $k = |S(v)|$, then the corresponding schedule for $v$ will be given as  $\s(v) = ((v_1, t_1), \ldots, (v_k, t_k))$, with $t_i \in \mathbb{R}_{\geq 0}$ denoting the starting time of the query for $v_i$. We will call $\s(v)$ the \emph{schedule of node} $v$. We will call a schedule assignment $\s$ \emph{consistent} with respect to search in a given tree $T$ if the following conditions are fulfilled:
\begin{enumerate}
\item[(i)] No two jobs in the schedule of a node overlap: for all $v \in V$, for two distinct jobs $(u_1,t_1), (u_2,t_2) \in \s(v)$, we have $|[t_1, t_1+w(u_1)] \cap [t_2, t_2+w(u_2)]| = 0$.
\item[(ii)] If $v$ is the parent of $v'$ in $T$ and $(u,t) \in \s(v')$, then we either also have $(u,t) \in \s(v)$, or the job $(v,t_v) \in \s(v)$ completes before the start of job $(u,t)$: $t_v + w(v) \leq t$.
\end{enumerate}
It follows directly from the definition that a consistent schedule assignment (and the underlying query sequence assignment) is uniquely determined by the collection of jobs $\{(v, t_v) : (v, t_v) \in \s(u), u \in V\}$. Note that not every vertex has to contain a query to itself in its schedule; we will occasionally write $t_v = \perp$ to denote that such a job is missing. \InJournal{ In this case, the jobs of all children of $v$ have to be contained in the schedule of node $v$.}


By extension of notation for sequence assignments, we will denote a strategy following a consistent schedule assignment $\s$ (i.e., executing the query jobs of schedule $\s$ at the prescribed times) as $\cA_{\s}$. We will then have:
\mathx
\cost_{\cA_{\s}}(T) = |\s|,
\mathx
where $|\s|$ is the \emph{duration} of schedule assignment $\s$, given as:
\mathx
|\s| = \max_{v\in V} |\s(v)|,
\mathx
with:
\mathx
|\s(v)| = \max_{(u,t)\in \s(v)} (t + w(u)).
\mathx

We remark that there always exists an optimal search strategy which is based on a consistent schedule. By a well-known characterization (cf.\ e.g.~\cite{Dereniowski06}), tree $T$ satisfies $\opt{T}=\tau \in \mathbb R$ if and only if there exists an assignment $I : V \to \mathcal I_\tau$ of intervals of time to nodes before deadline $\tau$, $\mathcal I_\tau  =\{[a,b] : 0 \leq a < b \leq \tau\}$, such that $|I(v)|=w(v)$ and if $|I(u) \cap I(v)| > 0$ for any pair of nodes $u, v \in V$, then the $u-v$ path in $T$ contains a separating vertex $z$ such that $\max I(z) \leq \min(I(u)\cup I(v))$. The corresponding schedule assignment of duration $\tau$ is obtained by adding, for each node $u \in V$, the job $(u, \min I(u))$ to the schedule of all nodes on the path from $u$ towards the root, until a node $v$ such that $\max I(v) \leq \min I(u)$
is encountered on this path. The consistency and correctness of the obtained schedule is immediate to verify.

\begin{observation}
For any tree $T$, there exists a query sequence assignment $S$ and a corresponding consistent schedule $\s$ on $T$ such that $|\s| = \opt T$.
\eop
\end{observation}

\section{The Results}
\subsection{\texorpdfstring{$(1+\eps)$}{(1+eps)}-Approximation  in \texorpdfstring{$n^{O(\log n/\eps^2)}$}{n\^{}O(log n/eps\^{}2)} Time} \label{sec:qptas}
We first present an approximation scheme for the weighted tree search problem with $n^{O(\log n)}$ running time. The main difficulty consists in obtaining a constant approximation ratio for the problem with this running time; we at once present this approximation scheme with tuned parameters, so as to achieve $(1+\eps)$-approximation in $n^{O(\log n/\eps^2)}$ time.

Our construction consists of two main building blocks. First, we design an algorithm based on a bottom-up (dynamic programming) approach, which considers exhaustively feasible sequence assignments and query schedules over a carefully restricted state space of size $n^{O(\log n)}$ for each node. The output of the algorithm provides us both with a lower bound on $\opt T$, and with a sequence assignment-based strategy $\cA_S$ for solving the tree search problem. The performance of this strategy $\cA_S$ is closely linked to the performance of $\opt T$, however, there is one type of query, namely a query on a vertex  
of small weight leading to a `down' response, due to whose repeated occurrence the eventual cost difference between $\cost_{\cA_S}(T)$ and $\opt T$ may eventually become arbitrarily large. To alleviate this difficulty, we introduce an alternative measure of cost which compensates for the appearance of the disadvantageous type of query.



We start by introducing some additional notation.
Let $\omega \in \mathbb R_+$, be an arbitrarily fixed value of weight and let $c\in\N$. The choice of constant $c \in \N$ will correspond to an approximation ratio of $(1+\eps)$ of the designed scheme for $\eps = 168/c$.

We say that a query to a vertex $v$ is a \emph{light down query} in some strategy if $w(v)<c\omega$ and $x\in V(T_v)$, i.e., it is also a `down' query, where $x$ is the target vertex.

For any strategy $\cA$, we denote by $\costo_\cA(T,x)$ its modified cost of finding target $x$, defined as follows.
Let $d_x$ be the number of light down queries when searching for $x$:
\mathx
d_x=\left|\{i : w(\Q_{\cA,i}(T,x))< c\omega \textnormal{ and } x \in V(T_{\Q_{\cA,i}(T,x)})\}\right|.
\mathx
 Then, the modified cost $\costo_\cA(T,x)$ is:
\begin{equation} \label{eq:costo-def}
\costo_\cA(T,x) = \cost_\cA(T,x) - (2c+1)\omega d_x.
\end{equation}
and by a natural extension of notation:
\mathx
\costo_\cA(T)=\max_{x\in V}\costo_\cA(T,x).
\mathx



The technical result which we will obtain in Section~\ref{sec:dp} may now be stated as follows.

\begin{proposition}\label{pro:dp}
For any $c \in \N$, $L \in \N$, there exists an algorithm running in time $(cn)^{O(L)}$, which for any tree $T$ constructs a stable sequence assignment $S$ and computes a value of $\omega$ such that $\omega \leq \frac {1}{L}\costo_{\cA_S}(T)$ and:
\mathx
\costo_{\cA_S}(T) \leq \left(1+\frac{12}{c} \right) \opt {T}.
\mathx
\end{proposition}

In order to convert the obtained strategy $\cA_S$ with a small value of $\costo$ into a strategy with small $\cost$, we describe in Section~\ref{sec:blackbox} an appropriate strategy conversion mechanism. The approach we adopt is applicable to any strategy based on a stable sequence assignment and consists in concatenating, for each vertex $v\in V$, a prefix to the query sequence $S(v)$ in the form of a separately computed sequence $R(v)$, which does not depend on $S(v)$. The considered query sequences are thus of the form $R(v) \circ S(v)$, where the symbol ``$\circ$'' represents sequence concatenation. Intuitively, the sequences $R$, taken over the whole tree, reflect the structure of a specific solution to the unweighted tree search problem on a contraction of tree $T$, in which each edge connecting a node to a child with weight at least $c\omega$ is contracted. We recall that the optimal number of queries to reach a target in an unweighted tree is $O(\log n)$, and the goal of this conversion is to reduce the number of light down queries in the combined strategy to at most $O(\log n)$.

\begin{proposition} \label{prop:cost-prime}\label{pro:cost-prime}
For any fixed $\omega > 0$ there exists a polynomial-time algorithm which for a tree $T$ computes a sequence assignment $R : V \to V^*$, such that, for any strategy $\cA_S$ based on a stable sequence assignment $S$, the sequence assignment $S^+$, given by $S^+(v) = R(v) \circ S(v)$ for each $v\in V$, has the following property:
$$\cost_{\cA_{S^+}}(T) \leq \costo_{\cA_{S}}(T) + 4\comega \log_2 n.$$
\end{proposition}

The proof of Proposition~\ref{pro:cost-prime} is provided in Section~\ref{sec:blackbox}.

We are now ready to put together the two bounds. Combining the claims of Proposition~\ref{pro:dp} for $L = \lceil c^2 \log_2 n \rceil$ (with $\omega \leq \frac {1}{L}\costo_{\cA_S}(T) \leq \frac{\costo_{\cA_{S}}(T)}{c^2 \log_2 n}$) and Proposition~\ref{pro:cost-prime}, we obtain:

\begin{align*}
\cost_{\cA_{S^+}}(T) &\leq \costo_{\cA_{S}}(T) + 4\comega \log_2 n \leq \costo_{\cA_{S}}(T) + 12 c\omega \log_2 n \
\leq\\&\leq
\costo_{\cA_{S}}(T) + 12 c \log_2 n \frac{\costo_{\cA_{S}}(T)}{c^2 \log_2 n}
\leq
\left(1+\frac{12}{c}\right) \costo_{\cA_{S}}(T) \leq\\
&\leq \left(1+\frac{12}{c}\right)^2 \opt {T} \leq \left(1+\frac{168}{c}\right) \opt {T}.
\end{align*}

After putting $\eps = \frac{168}{c}$ and noting that in stating our result we can safely assume $c = O(\mathrm{poly}(n))$ (beyond this, the tree search problem can be trivially solved optimally in $O(n^n)$ time using exhaustive search), we obtain the main theorem of the Section.

\begin{theorem}\label{thm:qptas}
There exists an algorithm running in $n^{O\left(\frac{\log n}{\eps^2}\right)}$ time, providing a $(1+\eps)$-approximation solution to the weighted tree search problem for any $0 < \eps < 1$.\eop
\end{theorem}

\subsection{Extension: A Poly-Time \texorpdfstring{$O(\sqrt{\log n})$}{O(sqrt(log n))}-Approximation Algorithm} \label{sec:approx}

We now present the second main result of this work. By recursively applying the previously designed QPTAS (Theorem~\ref{thm:qptas}), we obtain a polynomial-time $O(\sqrt{\log n})$-approximation algorithm for finding search strategy for an arbitrary weighted tree.
We start by informally sketching the algorithm --- we follow here the general outline of the idea from \cite{CicaleseKLPV14}.
The algorithm is recursive and starts by finding a minimal subtree $T^*$ of an input tree whose removal disconnects $T$ into subtrees, each of size bounded by $n/2^{\sqrt{\log n}}$.
The tree $T^*$ will be processed by our optimal algorithm described in Section~\ref{sec:qptas}.
This results either in locating the target node, if it belongs to $T^*$, or identifying the component of $T-T^*$ containing the target, in which case the search continues recursively in the component.
However, for the final algorithm to have polynomial running time, the tree $T^*$ needs to be of size $2^{O({\sqrt{\log n}})}$.
This is obtained by contracting paths in $T^*$ (each vertex of the path has at most two neighbors in $T^*$) into single nodes having appropriately chosen weights.
Since $T^*$ has $2^{O({\sqrt{\log n}})}$ leaves, this narrows down the size of $T^*$ to the required level  and we argue that an optimal search strategy for the `contracted' $T^*$ provides a search strategy for the original $T^*$ that is within a constant factor from the cost of $T^*$.

A formal exposition and analysis of the obtained algorithm is provided in \InJournal{Section~\ref{sec:algosqrt}.}\InConference{the Appendix.}
\begin{theorem} \label{thm:recursive-algo}
There is a $O(\sqrt{\log n})$-approximation polynomial time algorithm for the weighted tree search problem.
\end{theorem}

\section{\InJournal{Proof of Proposition~\ref{pro:dp}: }Quasi-Polynomial Computation of Strategies with Small \texorpdfstring{$\costo$}{Modified Cost}}\label{sec:dp}

\subsection{Preprocessing: Time Alignment in Schedules}

We adopt here a method similar but arguably more refined than rounding techniques in scheduling problems of combinatorial optimization, showing that we could discretise the starting and finishing time of jobs, as well as weights of vertices, in a way to restrict the size of state space for each node to $n^{O(\log n)}$, without introducing much error.

Fix $c \in \N$ and $\omega = \frac{a}{cn}$ for some $a\in \N$. (In subsequent considerations, we will have $c = \Theta(1/\eps)$, $a=O(\frac{n}{\log n})$ and $\omega = \Omega(\eps/\log n)$.) Given a tree $T = (V,E,w)$, let $T' = (V,E,w')$ be a tree with the same topology as $T$ but with weights rounded up as follows:
\begin{equation} \label{eq:wprime-def}
w'(v)=
\begin{cases}
\roundup{w(v)}{\omega}, &\text{if }w(v)>c\omega,\\
\roundup{w(v)}{\frac{1}{cn}}, &\text{otherwise.}
\end{cases}
\end{equation}
We will informally refer to vertices with $w(v)>c\omega$ (equivalently $w'(v)>c\omega$) as \emph{heavy vertices} and vertices with $w(v) \leq c\omega$ (equivalently $w'(v)\leq c \omega$) as \emph{light vertices}.
(Note that $w(v)\leq c\omega$ if and only if $w'(v)\leq c\omega$.)
When designing schedules, we consider time divided into \emph{boxes} of duration $\omega$, with the $i$-th box  equal to $[i\omega, (i+1) \omega]$. Each box is divided into $a$ identical \emph{slots} of length $\frac 1 {cn}$.

In the tree $T'$, the duration of a query to a heavy
vertex is an integer number of boxes, and the duration of a query to a light vertex is an integer number of
slots. We next show that, without affecting significantly the approximation ratio of the strategy, we can align each query to a heavy vertex in the schedule so that it occupies an interval of full adjacent boxes, and each query to a light
vertex in the schedule
so that it occupies an interval of full adjacent slots (possibly contained in more than one box).

We start by showing the relationship between the costs of optimal solutions for trees $T$ and $T'$.
\begin{lemma} \label{lem:costTTprime}
$\opt T \leq \opt {T'}\leq (1+\frac{2}{c})\opt T$.
\end{lemma}
\InJournal{
\begin{proof}

The inequality $\opt{T}\leq \opt{T'}$ follows directly from the monotonicity of the cost of the solution with respect to vertex weights, since we have $w'(v) \geq w(v)$, for all $v\in V$. 

To show the second inequality, we note that by the definition of weights~\eqref{eq:wprime-def}, for any vertex $v$, $w'(v) \leq (1+\frac{1}{c})w(v) + \frac{1}{cn}$.

Consider an optimal strategy $\cO$ for tree $T$ and let $\Q_{\cO}(T,x) = (v_1, \ldots, v_k)$ be the time-ordering of queries performed by strategy $\cO$ on tree $T$ to find a target vertex $x$. Let $\cO'$ be the strategy which follows the same time-ordering of queries when locating target $x$ in $T'$. We have:
\begin{align*}
\cost_{\cO'}(T', x) & = \sum_{i=1}^k w'(v_i) \leq \sum_{i=1}^k \left(\left(1+\frac{1}{c}\right)w(v) + \frac{1}{cn}\right) \leq \frac{1}{c} + \left(1+\frac{1}{c}\right) \sum_{i=1}^k w(v) \leq \\
& \leq \left(1+\frac{2}{c}\right) \opt T,
\end{align*}
where we used the fact that, by Observation~\ref{obs:opt-bounds}, $\opt T \geq 1$. Since $\opt {T'} \leq \max_{x\in V}\cost_{\cO'}(T', x) $, the claim follows.
\end{proof}
}

\begin{lemma} \label{lem:rounding-eps}
There exists a consistent schedule assignment $\s$ for tree $T'$ such that $ \cost_{\cA_{\s}} (T')\leq (1+\frac{3}{c})\opt{T'}$ and for all $v\in V$ we have that
\begin{itemize}
\item if $w'(v)> c\omega$, ($v$ is heavy), then the starting time $t$ of any job $(v,t)$ in the schedule $\s(u)$ of any $u \in V$ is an integer multiple of $\omega$ (aligned to a box),
\item
if $w'(v)\leq c\omega$, ($v$ is light), then
the starting time $t$ of any query $(v,t)$ in the schedule $\s(u)$ of any $u \in V$ is an integer multiple of $\frac 1{cn}$ (aligned to a slot).
\end{itemize}
\end{lemma}
\InJournal{
\begin{proof}
We consider an optimal consistent schedule assignment $\hat \Sigma$ for tree $T'$, $|\hat \Sigma|=\opt{T'}$. Fix $u \in V$ arbitrarily, and let $(v_{u,i}, t_{u,i})$ be the $i$-th query job in $\hat \Sigma (u)$. Consider now the schedule $\hat \Sigma^* (u)$ for $T$ based on the same sequence assignment, in which the job $(v_{u,i}, t_{u,i})$ is replaced by the job $(v_{u,i}, t^*_{u,i})$ with $t^*_{u,i} = (1+\frac{2}{c}) t_{u,i}$. We have for any two consecutive jobs at $u$:
\begin{equation} \label{eq:tstar-diff}
t^*_{u,{i+1}} - t^*_{u,i} = \left(1+\frac{2}{c}\right) (t_{u,{i+1}} - t_{u,{i}}) \geq \left(1+\frac{2}{c}\right)w(v_{u,i}),
\end{equation}
where we assume by convention that for the last job index $i_{max}$, $t_{u,{i_{max}+1}} = |\hat \Sigma(u)|$.
We now observe that schedule assignment $\hat\Sigma^*$ on tree $T$ can be directly converted into schedule assignment $\s$ on tree $T'$ as follows. The query sequence of each vertex is preserved unchanged. If $v_{u,i}$ is a heavy vertex, then within time interval $[t^*_{u,{i}}, t^*_{u,{i+1}}]$ we allocate to vertex $v_{u,i}$ an interval of full boxes, starting at time $\lceil t^*_{u,i} \rceil_{\omega}$. Indeed, by \eqref{eq:tstar-diff} we have:
$$
t^*_{u,{i+1}} -  \lceil t^*_{u,i} \rceil_{\omega} > t^*_{u,{i+1}} - t^*_{u,i} - \omega > \left(1+\frac{2}{c}\right)w(v_{u,i}) - \omega >
w(v_{u,i}) + \omega > w'(v_{u,i}).
$$
Since no two jobs overlap and the time transformation is performed identically for all vertices, the validity and consistency of schedule assignment $\s$ for tree $T'$ follows. We also have $|\s| \leq (1+\frac{2}{c}) |\hat \Sigma| =  (1+\frac{2}{c}) \opt {T'}$.

To obtain the second part of the claim (alignment for light vertices) it suffices to round up the starting time of query times of all (light) vertices to an integer multiple of $\frac{1}{cn}$. Since all weights in $T'$ are integer multiples of $\frac{1}{cn}$, and so are the starting times of queries to heavy vertices in $\s$, the correctness and consistency of the obtained schedule again follows directly. This final transformation increases the duration by at most $\frac{1}{c} \leq \frac{1}{c} \opt {T'}$, and combining the bounds for both the transformations finally gives the claim.
\end{proof}
}

A schedule on tree $T'$ satisfying the conditions of Lemma~\ref{lem:rounding-eps}, and the resulting search strategy, are called \emph{aligned}. Subsequently, we will design an aligned strategy on tree $T'$, and compare the quality of the obtained solution to the best aligned strategy for $T'$.

The intuition between the separate treatment of heavy vertices (aligned to boxes) and light vertices (aligned to slots) in aligned schedules is the following. Whereas the time ordering of boxes is essential in the design of the correct strategy, in our dynamic programming approach we will not be concerned about the order of slots within a single box (i.e., the order of queries to light vertices placed in a single box). This allows us to reduce the state space of a node. Whereas the ordering of slots in the box will eventually have to be repaired to provide a correct strategy, this will not affect the quality of the overall solution too much (except for the issue of light down queries pointed out earlier, which are handled separately in Section~\ref{sec:blackbox}).

\subsection{Dynamic Programming Routine for Fixed Box Size}

Let the values of parameter $c$ and box size $\omega$ be fixed as before. Additionally, let $L \in \N$ be a parameter representing the time limit for the duration of the considered vertex schedules when measured in boxes, i.e., the longest schedule considered by the procedure will be of length $L \omega$ (we will eventually choose an appropriate value of $L = O(\log n)$ as required when showing Theorem~\ref{thm:qptas}).

\InConference{
In order to lower-bound the duration of the consistent aligned schedule assignment with minimum cost, we perform an exhaustive bottom-up evaluation of aligned schedules which satisfy constraints on the occupancy of slots. However, instead of considering individual slots of a schedule which may be empty or full, for reasons of efficiency we consider the \emph{load} $s_v[p]$ of each box, $0\leq p < L$, in the same schedule, defined informally as the proportion of the duration of the occupied slots within the box to the duration  $\omega$ of the box.  In the Appendix, we formally show the following claim.

\begin{lemma}\label{lem:merge-insert-box}
Assume that the data structure $(s_v,t_v)_{v\in V}$ corresponds to a consistent schedule. Let $v \in V$ be an arbitrarily chosen node with set of children $\{v_1, \ldots, v_l\}$.
Then the set of queried nodes forms an edge cover of the tree:
\begin{equation}
\label{con:vertexcover}
\text{If $t_v = \perp$, then $t_{v_j} \neq \perp$, for all $1\leq j \leq l$.}
\end{equation}
Moreover, let completion time $\tinserted{v}$ of the query to $v$ given as:
$$
\tinserted{v} = \begin{cases}
t_v + w'(v), & \text{if $t_v \neq \perp$,}\\
+\infty, & \text{if $t_v = \perp$.}\\
\end{cases}
$$
Let $a_p$ be the contribution to the load of the $p$-th time box of the query job for vertex $v$, i.e.
$$
a_p = \begin{cases}
\frac{1}{\omega}|[t_v,\tinserted{v}] \cap [p\omega, (p+1)\omega]| & \text{if $t_v \neq \perp$,}\\
0 & \text{if $t_v = \perp$.}
\end{cases}
$$
Then, for any box $[p\omega, (p+1)\omega]$, $0 \leq p < L$, we have the following bounds on the amount of load which can be packed into the box:
\begin{equation}
\label{con:merge-insert-box}
\begin{rcases}
& s_v[p] = a_p + \sum_{j=1}^l s_{v_j}[p] \in [0,1], && \text{when $\tinserted{v}\geq (p+1)\omega$
,}\\
& s_v[p] \geq a_p, && \text{when $p\omega  < \tinserted{v} < (p+1)\omega$
, }\\
& s_v[p] = 0, && \text{when $\tinserted{v} \leq p\omega$.}
\end{rcases}
\end{equation}
Moreover,  for any box $[p\omega, (p+1)\omega]$, $0 \leq p < L$, we have that the total load of a query to $v$ and queries propagated from any of the subtrees cannot exceed $1$:
\begin{equation}
\label{con:fullboxsafe}
\text{For all $1\leq j \leq l$, the following bound holds: $s_{v_j}[p] + a_p \leq 1$.}
\end{equation}
\end{lemma}
}
\InJournal{
Before presenting formally the considered quasi-polynomial time procedure, we start by outlining an (exponential time) algorithm which verifies if there exists an aligned schedule assignment $\hat \Sigma$ for $T'$ whose duration is at most $L\omega$. Notice that since all weights in $T'$ are integer multiples of $\frac{1}{cn}$, the optimal aligned schedule assignment will start and complete the execution of all queries at times which are integer multiples of $\frac{1}{cn}$; thus, we may restrict the considered class of schedules to those having this property. Any possible schedule of length at most $L \omega$ at a vertex $v$, which may appear in $\hat \Sigma$, will be represented in the form of the pair $(\sigma_v,t_v)$, where:
\begin{itemize}
\item $\sigma_v$ is a Boolean array with $L \omega cn$ entries, where $\sigma_v[i]=1$ when time slot $[\frac{i}{cn}, \frac{i+1}{cn}]$ is occupied in the schedule at $v$, and $\sigma_v[i]=0$ otherwise.
\item $t_v \in \mathbb R$ represents the start time of the query to $v$ in the schedule of $v$ (we put $t=\perp$ if such a query does not appear in the schedule).
\end{itemize}
We now state some necessary conditions for a consistent schedule, known from the analysis of the unweighted search problem (cf.~e.g.~\cite{IyerRV88,OnakP06,Schaffer89}). The first observation expresses formally the constraint that the same time slot cannot be used in the schedules of two children of a node $v$, unless it is separated by an (earlier) query to node $v$ itself. All time slots before the starting time $t_v$ of job $(v,t_v)$ are free if and only if the corresponding time slot is free for all of the children of $v$.

\begin{observation}\label{obs:merge-insert}
Assume that the tuple $(\sigma_v,t_v)_{v\in V}$  corresponds to a consistent schedule. Let $v \in V$ be an arbitrarily chosen node with set of children $\{v_1, \ldots, v_l\}$. Let the completion time $\tinserted{v}$ of the query to $v$ in the schedule of $v$ be given as:
$$
\tinserted{v} = \begin{cases}
t_v + w'(v), & \text{if $t_v \neq \perp$,}\\
+\infty, & \text{if $t_v = \perp$.}\\
\end{cases}
$$
Then, for any time slot $[\frac{i}{cn}, \frac{i+1}{cn}]$, we have:
\begin{equation}
\label{con:merge-insert}
\begin{rcases}
& \sigma_v[i] = \sum_{j=1}^l \sigma_{v_j}[i], && \text{\quad when $\frac{i+1}{cn} \leq t_v$,}\\
& \sigma_v[i] = 1 \text{ and } \sum_{j=1}^l \sigma_{v_j}[i] = 0, && \text{\quad when $t_v < \frac{i+1}{cn} \leq \tinserted{v}$,}\\
& \sigma_v[i] = 0, && \text{\quad when $\frac{i+1}{cn} > \tinserted{v}$.}
\end{rcases}
\end{equation}
\end{observation}

We remark that the last of the above conditions~\eqref{con:merge-insert} follows from the w.l.o.g.\ assumption we made when defining sequence assignments that whenever node $v$ appears in the schedule of $v$, it is the last node in the query sequence for $v$.

Moreover, any valid search strategy which locates a target vertex must eventually query at least one of the endpoints of every edge of the tree $T'$, since otherwise, it will not be able to distinguish targets located at these two endpoints. We thus make the following observation.

\begin{observation}\label{obs:vertexcover}
Assume that the tuple $(\sigma_v,t_v)_{v\in V}$ represents a consistent schedule. Let $v \in V$ be an arbitrarily chosen node with set of children $\{v_1, \ldots, v_l\}$. Then:
\begin{equation}
\label{con:vertexcover}
\text{If $t_v = \perp$, then $t_{v_j} \neq \perp$, for all $1\leq j \leq l$.}
\end{equation}
\end{observation}

\MC{This should be mentioned earlier. Done.}

Conditions~\eqref{con:merge-insert} and~\eqref{con:vertexcover} provide us with necessary conditions which must be satisfied by any consistent aligned schedule assignment. 

In order to lower-bound the duration of the consistent aligned schedule assignment with minimum cost, we perform an exhaustive bottom-up evaluation of aligned schedules which satisfy the constraints of~\eqref{con:vertexcover}, and a slightly weaker form of the constraints of~\eqref{con:merge-insert}. These weaker constraints are introduced to reduce the running time of the algorithm. Instead of considering individual slots of a schedule which may be empty or full, $\sigma_v [i] \in \{0,1\}$, we consider the load of each box in the same schedule, defined as the proportion of occupied slots within the box. Formally, for the $p$-th box, $0\leq p < L$, the \emph{load} $s_v[p]$ is given as:
$$
s_v [p] = \frac{1}{\omega cn}\sum_{i = p \cdot\omega cn}^{(p+1)\omega cn -1} \sigma_v [i],\quad \quad s_v[p]\in \left\{0, \frac{1}{\omega cn}, \frac{2}{\omega cn}, \ldots, 1 \right\},
$$
where we recall that $\omega cn$ is an integer by the choice of $\omega$.
We will call a box with load $s_v [p] = 0$ an \emph{empty box}, a box with load $s_v [p] = 1$ a \emph{full box}, and a box with load $0 < s_v [p] < 1$ a \emph{partially full box} in the schedule of $v$.

By summing over all slots within each box, we obtain the following corollary directly from Observation~\ref{obs:merge-insert}.

\begin{corollary}\label{cor:merge-insert-box}
Assume that the tuple $(s_v,t_v)_{v\in V}$ corresponds to a consistent schedule. Let $v \in V$ be an arbitrarily chosen node with set of children $\{v_1, \ldots, v_l\}$ and completion time $\tinserted{v}$ of the query to $v$ given as in Observation~\ref{obs:merge-insert}. Let $a_p$ be the contribution to the load of the $p$-th box of the query job for vertex $v$, i.e.
$$
a_p = \begin{cases}
\frac{1}{\omega}|[t_v,\tinserted{v}] \cap [p\omega, (p+1)\omega]| & \text{if $t_v \neq \perp$,}\\
0 & \text{if $t_v = \perp$.}
\end{cases}
$$
Then, for any box $[p\omega, (p+1)\omega]$, $0 \leq p < L$, we have:
\begin{equation}
\label{con:merge-insert-box}
\begin{rcases}
& s_v[p] = a_p + \sum_{j=1}^l s_{v_j}[p] \in [0,1], && \text{when $\tinserted{v}\geq (p+1)\omega$
,}\\
& s_v[p] \geq a_p, && \text{when $p\omega  < \tinserted{v} < (p+1)\omega$
, }\\
& s_v[p] = 0, && \text{when $\tinserted{v} \leq p\omega$.}
\end{rcases}
\end{equation}
Moreover, for any box $[p\omega, (p+1)\omega]$, $0 \leq p < L$, we have:
\begin{equation}
\label{con:fullboxsafe}
\text{For all $1\leq j \leq l$, the following bound holds: $s_{v_j}[p] + a_p \leq 1$.}
\end{equation}
\end{corollary}
We remark that the statement of Corollary~\ref{cor:merge-insert-box} treats specially one box, namely the one which contains strictly within it 
the time moment $\tinserted{v}$. For this box, we are unable to make a precise statement about $s_v[p]$ based on the description of the schedules of its children, and content ourselves with a (potentially) weak lower bound $s_v[p] \geq a_p = \frac{1}{\omega}(\tinserted{v} - p\omega)$. This is the direct reason for the slackness in our subsequent estimation, which loses $\omega$ time per down query. However, we note that by the definition of aligned schedule, a query to a heavy vertex will never begin or end strictly inside a box, and will not lead to the appearance of this issue. 
We remark that condition~\eqref{con:fullboxsafe} additionally stipulates that within any box, it must be possible to schedule the contribution of the query to $v$ and the contribution of any child $v_j$ to the load of the box in a non-overlapping way.
}
We now show that the shortest schedule assignments satisfying the set of constraints~\eqref{con:vertexcover}, \eqref{con:merge-insert-box}, and~\eqref{con:fullboxsafe} can be found in $n^{O(\log n)}$ time. This is achieved by using the procedure $\textsc{BuildStrategy}$, presented in Algorithm~\ref{alg:dp}, which returns for a node $v$ a non-empty set of schedules $\mathcal{\hat S}[v]$, such that each $s_v \in \mathcal{\hat S}[v]$ can be extended into the sought assignment of schedules in its subtree, $(s_u, t_u)_{u\in V(T_v)}$. In the statement of Algorithm~\ref{alg:dp}, we recall that, given a tree $T = (V,E,w)$, tree $T' = (V,E,w')$ is the tree with weights rounded up to the nearest multiple of the length of a slot (see Equation~\eqref{eq:wprime-def}).

The subsequent steps taken in procedure $\textsc{BuildStrategy}$ can be informally sketched as follows.
The input tree $T'$ is processed in a bottom-up manner and hence, for an input vertex $v$, the recursive calls for its children $v_1,\ldots,v_l$ are first made, providing schedule assignments for the children (see lines \ref{l:rec-beg}-\ref{l:rec-end}).
Then, the rest of the pseudocode is responsible for using these schedule assignments to obtain all valid schedule assignments for $v$.
Lines \ref{l:merge-beg}-\ref{l:merge-end} merge the schedules of the children in such a way that a set ${\hat S}^*_{i}$, $i\in\{1,\ldots,l\}$, contains all schedule assignments computed on the basis of the schedules for the children $v_1,\ldots,v_i$.
Thus, the set ${\hat S}^*_{l}$ is the final product of this part of the procedure and is used in the remaining part.
Note that a schedule assignment in ${\hat S}^*_{l}$ may not be valid since a query to $v$ is not accommodated in it --- the rest of the pseudocode is responsible for taking each schedule $s\in{\hat S}^*_{l}$ and inserting a query to $v$ into $s$.
More precisely, the subroutine $\textsc{InsertVertex}$ is used to place the query to $v$ at all possible time points (depending whether $v$ is heavy or light).
We note that the subroutine \textsc{MergeSchedules}, for each schedule $s$ it produces, sets a Boolean `flag' $s.must\_contain\_v$ that whenever equals $false$, indicates that querying $v$ is not necessary in $s$ to obtain a valid schedule for $v$ (this happens if $s$ queries all children of $v$).
A detailed analysis of procedure \textsc{BuildStrategy} can be found in \InJournal{the proof of Lemma~\ref{lem:dplower1}}\InConference{the Appendix}.

\begin{algorithm}[t]
\small
\caption{Dynamic programming routine \textsc{BuildStrategy} for a tree $T'$. $L, c\in \N$ are global parameters. Subroutines \textsc{MergeSchedules} and \textsc{InsertVertex} are provided \InJournal{further on.}\InConference{in the Appendix.}}\label{alg:dp}
\begin{algorithmic}[1]
\Procedure{BuildStrategy}{vertex $v$, box size $\omega \in \mathbb R$}
\State $l \gets$ number of children of $v$ in $T'$ \quad// Denote by $v_1,\ldots,v_l$ the children of $v$.
\For{$i=1..l$ \label{l:rec-beg}}
	\State $\mathcal{\hat S}[v_i] \gets$ \Call{BuildStrategy}{$v_i$, $\omega$}; \label{l:rec-end}
\EndFor
\State $s \gets 0^L$
\State $s.max\_child\_load \gets 0^L$
\State $s.must\_contain\_v \gets false$
\State $\mathcal{\hat S}_{0} \gets  \{ s \}$ \quad// $\mathcal{\hat S}_{0}$ contains the schedule with no queries.
\State // Inductively, $\mathcal{\hat S}^*_{i}$ is based on merging schedules at $v_1,\ldots,v_{i}$.
\For{$i=1..l$ \label{l:merge-beg}}
    \State $\mathcal{\hat S}^*_i \gets \emptyset$
    \For{each schedule $s \in \mathcal{\hat S}^*_{i-1}$}
        \For{each schedule $s_{add} \in \mathcal {\hat S} [v_i]$}
            \State $\mathcal{\hat S}^*_i \gets \mathcal{\hat S}^*_i \cup $  \Call{MergeSchedules}{$s$, $s_{add}$, $\omega$}; \label{l:merge-end}
    	\EndFor
    \EndFor
\EndFor
\State $\mathcal{\hat S}[v] \gets \emptyset$
\For{each $s \in \mathcal{\hat S}^*_l$}
    \If{$w'(v) > c \omega$} \quad // $v$ is heavy
        \For{$p = 0..L-1$} \quad //attempt to insert (into $s$) query to $v$ starting from \mbox{time-box $p$}
            \State $\mathcal{\hat S}[v] \gets \mathcal{\hat S}[v] \cup$ \Call{InsertVertex}{$s, v, \omega, p\cdot \omega$}
        \EndFor
    \Else  \quad //$v$ is light
        \For{real $t = 0..L\cdot \omega$ \textbf{step} $\frac{1}{cn}$}
            \State //attempt to insert (into $s$) query to $v$ at a slot from time $t$
            \State $\mathcal{\hat S}[v] \gets \mathcal{\hat S}[v] \cup$ \Call{InsertVertex}{$s, v, \omega, t$}
        \EndFor
    \EndIf
    \If{$s.must\_contain\_v = false$}
        \State $\mathcal{\hat S}[v] \gets \mathcal{\hat S}[v] \cup$ \Call{InsertVertex}{$s, v, \omega, \perp$}
    \EndIf
\EndFor

\State \textbf{return} $\mathcal{\hat S}[v]$
\EndProcedure
\end{algorithmic}
\end{algorithm}

\InJournal{\begin{algorithm*}[t]
\small
\caption{Subroutines \textsc{MergeSchedules} and \textsc{InsertVertex}  of procedure \textsc{BuildStrategy} from Algorithm~\ref{alg:dp}.}\label{alg:dp-sub}
\begin{algorithmic}[1]
\Procedure{MergeSchedules}{schedule $s_{orig}$, schedule $s_{add}$, box size $\omega \in \mathbb R$}
    \State $s \gets s_{orig}$ \quad // copy schedule and its properties to answer
    \For{$p = 0..L-1$} \quad // for each time-box add load of $s_1$ and $s_2$
        \State $s[p] \gets s_{orig}[p] + s_{add}[p]$
        \If{$s[p] > 1$}
            \State {$s[p] \gets +\infty$}
        \EndIf
        \State $s.max\_child\_load[p] \gets \max\{s.max\_child\_load[p], s_{add} [p]\}$
    \EndFor
    \If{$s_{add}.t_v = \perp$}
        \State $s.must\_contain\_v \gets true$
    \EndIf
	\State \textbf{return} $s$
\EndProcedure
\\

\Procedure{InsertVertex}{schedule $s_{orig}$, vertex $v$, box size $\omega \in \mathbb R$, time $t\in \mathbb R \cup \{\perp\}$}
	\State $s \gets 0^L$ \quad // initialize empty schedule for answer
    \If{$t \neq \perp$}
        \State $I \gets [t, t+w'(v)]$ \quad // time interval into which query to $v$ is being inserted
        \State $s.t_v \gets t$
        \State $\tinserted{v} \gets t + w'(v)$
    \Else
        \State $I \gets \emptyset$
        \State $s.t_v \gets \perp$
        \State $\tinserted{v} \gets +\infty$
    \EndIf
    \For{$p = 0..L-1$} \quad // for each time-box
    	\State $a_p \gets \frac{1}{\omega}|I \cap [p \cdot \omega, (p+1)\cdot \omega]|$ \quad // contribution of query to $v$ to load of box $p$
        \If {$s.max\_child\_load [p] + a_p>1$}
            \State \textbf{return} $\emptyset$
        \EndIf
        \If {$\tinserted{v}\geq(p+1)\omega$
        } \quad
            \State $s[p] \gets s_{orig}[p] + a_p$ // add load from children in box $p$
            \If {$s[p] > 1$} \quad //insertion failed
                \State \textbf{return} $\emptyset$
            \EndIf
        \Else
            \State $s[p] \gets a_p$
        \EndIf
    \EndFor
	\State \textbf{return} $\{s\}$
\EndProcedure
\end{algorithmic}
\end{algorithm*}
}

\begin{lemma}\label{lem:dplower1}
For fixed constants $L, c\in \N$, calling procedure $\textsc{BuildStrategy}(r(T),\omega)$, where $r(T)$ is the root of the tree, determines if there exists a tuple $(s_v,t_v)_{v\in V}$ which satisfies constraints~\eqref{con:vertexcover}, \eqref{con:merge-insert-box}, and~\eqref{con:fullboxsafe}, or returns an empty set otherwise.
\end{lemma}
\InJournal{
\begin{proof}
The formulation of procedure $\textsc{BuildStrategy}$ directly enforces that the constraints~\eqref{con:vertexcover}, \eqref{con:merge-insert-box}, and~\eqref{con:fullboxsafe} are fulfilled at each level of the tree, in a bottom-up manner.

For each vertex $v \in V$, we show by induction on the tree size that upon termination of procedure $\textsc{BuildStrategy}(v,\omega)$, the returned variable $\mathcal{\hat  S} [v]$ is the set of all \emph{minimal} schedules $(s_v, t_v) \in  \mathcal{\hat S} [v]$ which can be extended within the subtree $T_v$ to a data structure $(s_u, t_u)_{u \in V(T_v)}$, for some $ (s_u, t_u)\in \mathcal {\hat S} [u]$, $u \in V(T_v)$, in such a way that the conditions~\eqref{con:vertexcover}, \eqref{con:merge-insert-box}, and~\eqref{con:fullboxsafe} hold within subtree $T_v$. Here, minimality of a schedule is a trivial technical assumption, understood in the sense of the following very restrictive partial order: we say $(s_v, t_v) \leq (s'_v, t'_v)$ if $s_v[p] \leq s'_v[p]$ for all $0 \leq p \leq L-1$ and $t_v = t'_v$. (In the pseudocode, rather than write $(s_v, t_v)$ as a pair variable, we include $t_v$ within the structure $s_v$ as its special field $s_v.t_v$.)

The algorithm proceeds to merge together exhaustively all possible choices of schedules $ (s_{v_i}, t_{v_i})\in \mathcal {\hat S} [v_i]$ of all children $v_i$ of $v$, $1\leq i \leq l$. The merge is performed by computing, for any fixed choice $(s_{v_i}, t_{v_i})_{1\leq i \leq l}$, the combined load of each box in the resultant schedule $s$:
\begin{equation}\label{eq:sumload}
s[p] \gets \sum_{i=1}^l s_{v_i} [p],
\end{equation}
where, as a technicality, we also put $s[p] \gets +\infty$ whenever we obtain excessive load in a box ($s[p]>1$), as to avoid inflating the size of the state space and consequently, the running time of the algorithm. In Algorithm~\ref{alg:dp}, the computation of $s[p]$ through the sum~\eqref{eq:sumload} proceeds by a processing of successive children $v_i$, $1\leq i \leq l$, so that a schedule $s$ stored in the data structure $\mathcal{\hat S}^*_i$ represents $s[p] = \sum_{i=1}^i s_{v_j} [p]$. The summation of load is performed within the subroutine $\textsc{MergeSchedules}$, which merges a schedule $s_{orig} \in \mathcal{\hat S}^*_{i-1}$ with a schedule $s_{add} \in \mathcal{\hat  S} [v_i]$ to obtain the new schedule $s \in \mathcal{\hat S}^*_{i}$.

Eventually, the set of schedules $\mathcal{\hat S}^*_{l}$, obtained after merging the schedules of all children of $v$, contains an element $s$ satisfying~\eqref{eq:sumload}. Next, we test all possible values of $t_v \in \mathbb R \cup \{\perp\}$, which are feasible for an aligned schedule. These values depend on whether vertex $v$ is heavy or light, for which $t_v$ should represent the starting time of a box or slot, respectively.
Using procedure $\textsc{InsertVertex}$, we then set the load of each box following~\eqref{con:merge-insert-box}:
\begin{equation}\label{eq:insertmerge}
s_v[p] \gets\begin{cases}
a_p + \sum_{j=1}^l s_{v_j}[p], & \text{when $\tinserted{v}\geq(p+1)\omega$
,}\\
a_p, & \text{when $p\omega  < \tinserted{v} < (p+1)\omega$
, }\\
0, & \text{when $\tinserted{v} \leq p\omega$,}
\end{cases}
\end{equation}
where $a_p$ is defined as in~\eqref{con:merge-insert}. In the pseudocode of function $\textsc{InsertVertex}$, for compactness we replace the second and third condition by equivalently setting $s_v[p] \gets a_p$ when the first condition does not hold. We additionally constrain in procedures $\textsc{MergeSchedules}$ and $\textsc{InsertVertex}$ the possibility of the condition $t_v = \perp$ occurring by enforcing the constraints of~\eqref{con:vertexcover} (corresponding of the setting of parameter $s.must\_contain\_v$ to $false$). Condition~\eqref{con:fullboxsafe} is enforced through procedures $\textsc{MergeSchedules}$ and $\textsc{InsertVertex}$ using the auxiliary array $s.max\_child\_load[p]$, $0\leq p \leq L-1$, defined so that $s.max\_child\_load[p] \gets \max_{1\leq j \leq l} s_{v_j}[p]$.

Since $\mathcal {\hat S} [v_i]$, for all $1\leq i \leq l$, contains all minimal schedules satisfying~\eqref{con:vertexcover}, \eqref{con:merge-insert-box}, and~\eqref{con:fullboxsafe}, the same holds for $\mathcal {\hat S} [v]$, which was constructed by enforcing only the required constraints. We remark that we obtain only the set of minimal (and not all) schedules due to the slight difference between~\eqref{eq:insertmerge} and~\eqref{con:merge-insert-box} in the second condition: instead of requiring $s_v[p] \geq a_p$, we put $s_v[p] \gets a_p$, thus setting the $p$-th coordinate of the schedule at its minimum possible value.
\end{proof}
}

It follows directly from Lemma~\ref{lem:dplower1} that, for any value $\omega^*$, tree $T$ may only admit an aligned schedule assignment of duration at most $\omega^* L$ if a call to procedure $\textsc{BuildStrategy}\allowbreak(r(T),\omega^*)$ returns a non-empty set. Taking into account Lemmas~\ref{lem:costTTprime} and~\ref{lem:rounding-eps}, we directly obtain the following lower bound on the length of the shortest aligned schedule in tree $T'$.

\begin{lemma}\label{lem:lb}
If $\textsc{BuildStrategy}(r(T),\omega^*) = \emptyset$, then:\InConference{\\}
\mathx
\omega^* L < \left (1+\frac{3}{c}\right) \opt {T'}  \leq \left(1+\frac{3}{c}\right)\left(1+\frac{2}{c}\right) \opt T \leq \left(1+\frac{11}{c}\right) \opt T.
\mathx
\eop
\end{lemma}

\InJournal{
Finally, we bound the running time of procedure $\textsc{BuildStrategy}$.}
\begin{lemma}
The running time of procedure $\textsc{BuildStrategy} (r(T), \omega)$ is at most $O((cn)^{\gamma L})$, for some absolute constant $\gamma = O(1)$, for any $\omega \leq n$.
\end{lemma}
\InJournal{
\begin{proof}
The procedure $\textsc{BuildStrategy}$ is run recursively, and is executed once for each node of the tree. The time of each execution is upper-bounded, up to multiplicative factors polynomial in $n$, by the size of the largest of the schedule sets named $\mathcal{\hat S}[u]$, $u\in V$, or $\mathcal{\hat S}^*_i$, appearing in the procedure. We further focus only on bounding the size $|\mathcal{\hat S}|$ of the state space of distinct possible schedules in the $(s_v, t_v)$ representation. The array $s_v$ has size $L$, with each entry $s_v [p]$, $0\leq p \leq L-1$, taking one of the values $s_v [p] \in \{0, \frac{1}{\omega cn}, \frac{2}{\omega cn},\ldots,  1 \}$, where the size of the set of possible values is $\omega cn+1 \in \N$. Additionally, in some of the auxiliary schedules, the additional array field $s_v.max\_child\_load$ has length $L$, with each entry $s_v.max\_child\_load [p]$, $0\leq p \leq L-1$, likewise taking one of the values from the set  $\{0, \frac{1}{\omega cn}, \frac{2}{\omega cn},\ldots,  1 \}$. Finally, for the time $t_v$, we have: $t_v \in \{0, \omega, 2\omega, \ldots, (L-1)\omega, \perp\}$, where the size of the set of possible values is $L+1$.

Overall, we obtain:
$$
|\mathcal{\hat S}| \leq (L+1) \left(\omega cn+1\right)^{L} \left(\omega cn+1\right)^{L} \leq (L+1) \left(cn^2+1\right)^{2L} < (cn)^{L \gamma'},
$$
where $\gamma'>0$ is a suitably chosen absolute constant. Accommodating the earlier omitted multiplicative $O(\mathrm{poly}(n))$ factors in the running time of the algorithm, we get the claim for some suitably chosen absolute constant $\gamma > \gamma'$.
\end{proof}
}

\InJournal{
\subsection{Sequence Assignment Algorithm with Small \texorpdfstring{$\costo$}{Modified Cost}}
}

\InJournal{
The procedure for computing a sequence assignment $S$ which achieves a small value of $\costo$ is given in Algorithm~\ref{alg:S}.}
\InConference{To complete the proof of Proposition~\ref{pro:dp}, we can now provide a strategy which achieves a small value of $\costo$. }
This relies on procedure $\textsc{BuildStrategy}(r(T), \omega)$ as an essential subroutine, first determining the minimum value of $\omega = \frac{i}{cn}$, $i\in \N$, for which $\textsc{BuildStrategy}$ produces a schedule. \InJournal{Since the schedule of a parent node $v$ is based on an insertion of a query to $v$ into the schedules of its children, a standard backtracking procedure allows us to determine the representation $(s_v, t_v)_{v \in V}$ of the schedules of all nodes of the tree.}
\InConference{Details of the approach are provided in the Appendix.}

\InJournal{
\DD{I do not quite see why we reconstruct query assignment in lies 7-9 the way we do. The problematic part for me is that \textsc{BuildStrategy}, and more precisely \textsc{InsertVertex}, is taking care of the fact that if one vertex "hides" query to a descendant, then the descendant is not propagated upwards? So it seems that $C(v)$ can be constructed on the basis of $s_v$: according to the definition of a schedule, is is the same as sequence assignment modulo time stamps, so perhaps we can just say that the schedule computed by \textsc{BuildStrategy} has for each vertex $v$ a sequence $((v_1,t_1),\ldots,(v_k,t_k))$ and define $S(v):=(v_1,\ldots,v_k)$. This should potentially replace lines 5-10. What am I missing?
}
\AK{It is quite possible that you are right (you are definitely right in terms of high level intuition); to be honest, I do not remember all the boundary cases / issues with rounding up and down to multiples of omega, cutting partially covered jobs, etc. Perhaps it's OK to simplify it.}
\begin{algorithm}
\small
\caption{Construction of sequence assignment $S$}\label{alg:S}
\begin{algorithmic}[1]
\State $\omega \gets \frac{1}{cn}$
\While{\Call{BuildStrategy}{$r(T)$, $\omega$} $= \emptyset$}
	\State $ \omega \gets \omega + \frac{1}{cn}$
\EndWhile
\State $(s_v, t_v)_{v \in V} \gets $ schedule assignment of duration at most $L \omega$, satisfying constraints~\eqref{con:vertexcover}, \eqref{con:merge-insert-box}, and~\eqref{con:fullboxsafe},
\StatexIndent[1] reconstructed by backtracking through the sets $(\mathcal {\hat S}[v])_{v\in V}$ computed in the last call
\StatexIndent[1] to procedure \Call{BuildStrategy}{$r(T)$, $\omega$}.
\For {$v \in V$}
  \State $C(v) \gets \emptyset$
  \For {$u \in V(T_v)$}
    \If{there is no vertex $z\neq u$ on the path from $v$ to $u$ s.t.\ $t_z < \lfloor t_u +  w'(u) \rfloor_\omega + \omega$ \label{ln:Cvcond}}
      \State\label{ln:Cv} $C(v) \gets C(v) \cup \{(\lfloor t_u \rfloor_{\omega}, \lceil t_u + w'(u) \rceil_{\omega}, u)\}$
    \EndIf
  \EndFor
  \State\label{ln:Cvsort} $S(v) \gets$ sequence of vertices (third field) of $C(v)$ sorted in non-decreasing
  \StatexIndent[2] order, with tuples compared by first field, then second field, then third field.
\EndFor
\State \textbf{return} $(S(v))_{v\in V}$
\end{algorithmic}
\end{algorithm}

We start by observing in Algorithm~\ref{alg:S} that if a node $v$ is not queried ($t_v = \perp$), then all of the children of $v$ belong to the schedules produced by procedure $\textsc{BuildStrategy}$ following condition~\eqref{con:vertexcover}, and thus they will also appear in $S(v)$. This guarantees the validity of the solution.

\begin{lemma} \label{lem:finalalg-correct}
Algorithm~\ref{alg:S} returns a correct query sequence assignment $S$ for tree $T$.
\eop
\end{lemma}

For the purposes of analysis, we extend the notion of backtracking procedure  $\textsc{BuildStrategy}$ in a natural way, so that, for every node $v\in V$ and box $0 \leq p \leq L-1$, we describe precisely the contribution $c_v[p,u]$ of each vertex $u \in V(T_v)$ to the load $s_v[p]$. \InJournal{(See Fig.~\ref{fig:group-1} for an illustration.)}
\InJournal{
\begin{figure}
\begin{minipage}[b]{.5\linewidth}
\centering
\includegraphics[width=2in]{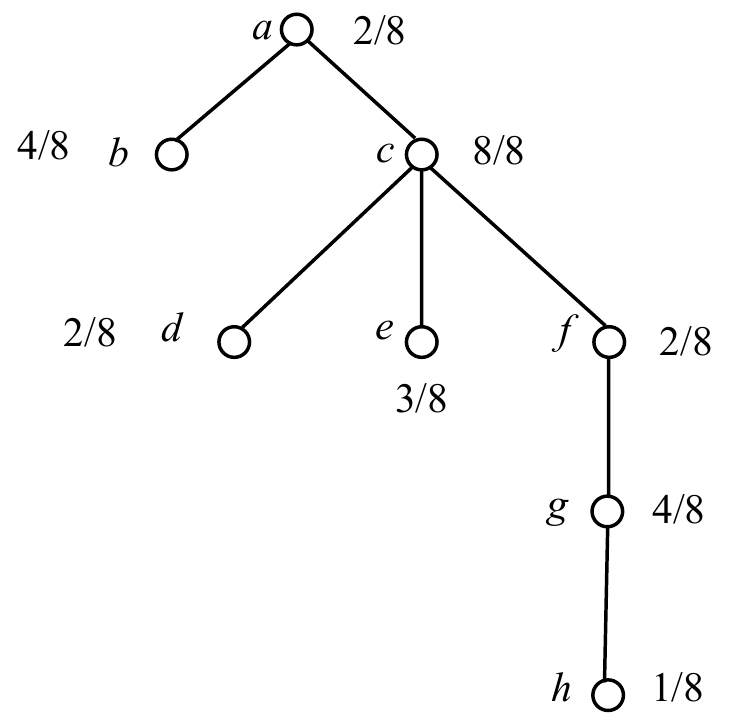}
\subcaption{Tree $T'$ with vertex weights.          \ \ \\ \ \ \\ \ \\ \ \\ \ }\label{fig:1a}
\end{minipage}%
\begin{minipage}[b]{.5\linewidth}
\centering
\includegraphics[width=3in]{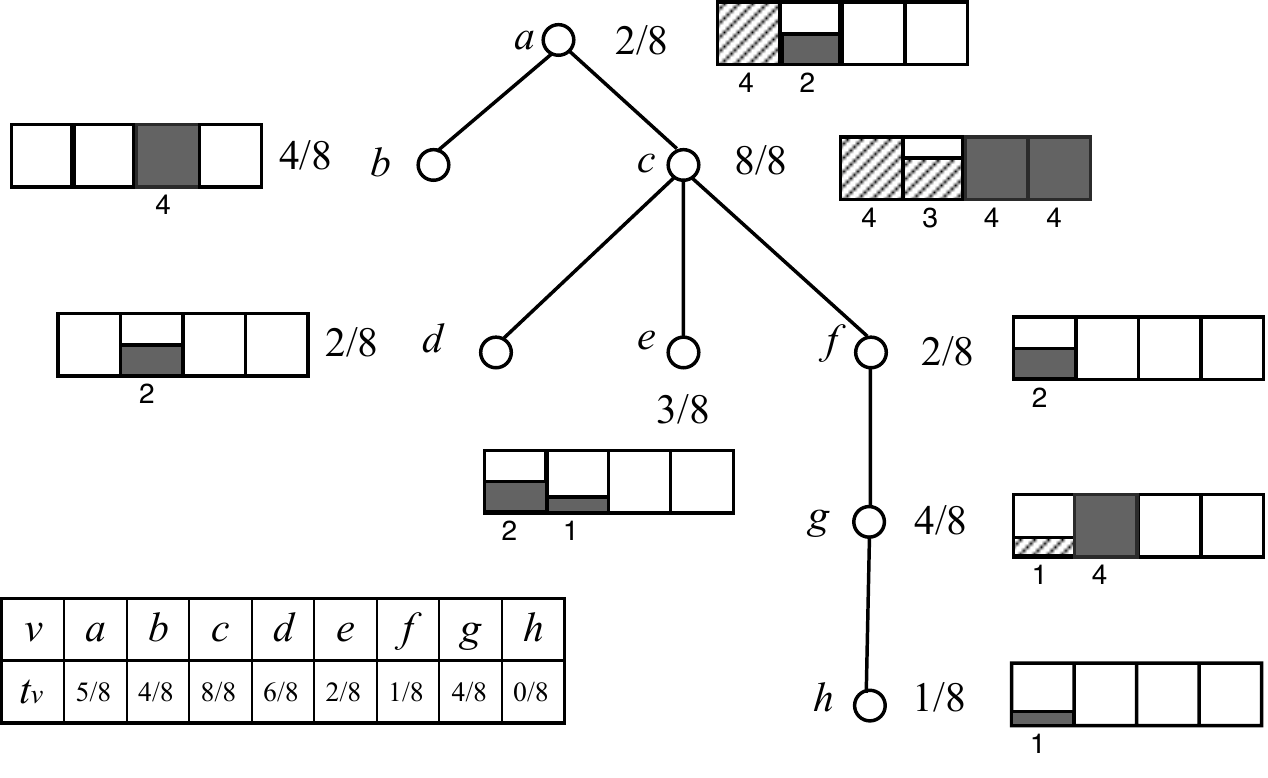}
\subcaption{Sample schedule $(s_v,t_v)_{v\in V}$ obtained by backtracking procedure BuildStrategy with parameters $c=1,n=8,\omega=\frac{4}{8}$, (box size $\frac{4}{8}$, slot size $\frac{1}{8}$, $4$ slots per box), $L=4$. Note that the schedules $(s_v)_{v\in V}$ may correspond to different starting times of jobs within the prescribed box; the provided $t_v$ are an example.}\label{fig:1b}
\end{minipage}

\begin{minipage}[b]{.5\linewidth}
\centering
\includegraphics[width=2.5in]{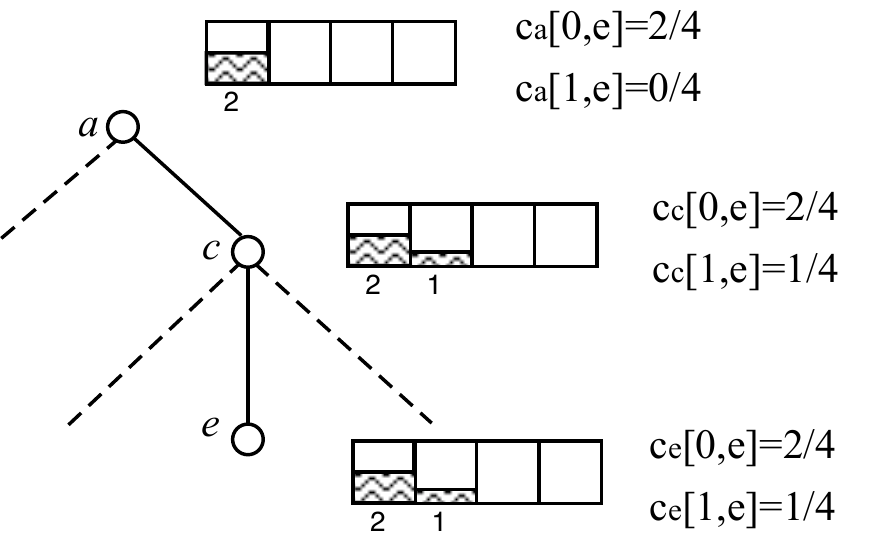}
\subcaption{Contribution of load of vertex $e$ to different vertices of the tree. \\ \ \\ \  }\label{fig:1c}
\end{minipage}
\begin{minipage}[b]{.5\linewidth}
\centering
\includegraphics[width=1.8in]{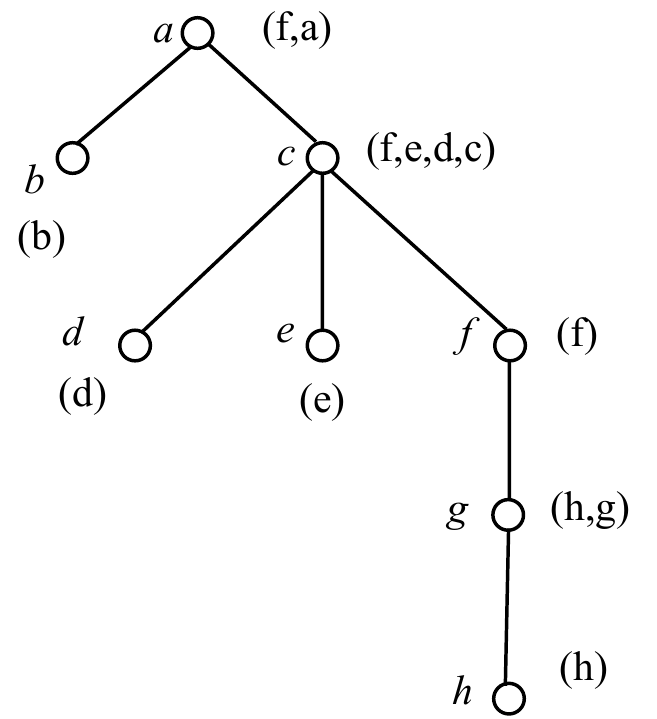}
\subcaption{Sequences $S(v)$ computed by Algorithm 3 based on provided $(s_v,t_v)_{v\in V}$. Note that vertex $e$ does not appear in $S(a)$ because of the query to $c$ on the way.}\label{fig:1d}
\end{minipage}

\caption{Illustration of Algorithm ~\protect\ref{alg:dp} and ~\protect\ref{alg:dp-sub}. The depicted tree $T'$ has vertex set $V=\{a,b,c,d,e,f,g,h\}$ and vertex weights:} \begin{tabular}{c|c|c|c|c|c|c|c|c}
$v$&a&b&c&d&e&f&g&h\\
\hline
 $w'(v)$&$\frac{2}{8}$&$\frac{4}{8}$&$\frac{8}{8}$&$\frac{2}{8}$&$\frac{3}{8}$&$\frac{2}{8}$&$\frac{4}{8}$&$\frac{1}{8}$\\
\end{tabular}\label{fig:1}\label{fig:group-1}

\end{figure}
}
Formally, for $u=v$ we have $c_v[p,v]  \gets a_p = |[t_v, t_v + w'(v)] \cap [p\omega, (p+1)\omega]|$ if $t_v\neq \perp$, and $c_v[p,v] \gets 0$, otherwise. Next, if $u \neq v$ and $u$ belongs to the subtree of child $v_i$ of $v$, we put:
$$
c_v[p,u]\gets\begin{cases}
c_{v_i}[p,u], & \text{if $\tinserted{v} > p \omega$,}\\
0, & \text{otherwise,}
\end{cases}
$$
where the insertion time $\tinserted{v}$ for $v$ is defined as in Observation~\ref{obs:merge-insert}. Comparing with~\eqref{eq:insertmerge}, we have directly for all $0 \leq p < L$:
$$
s_v [p] = \sum_{u \in V(T_v)} c_{v} [p,u].
$$
Let $p_s(u)$ and $p_f(u)$ be the indices of the starting and final box, respectively, to which  vertex $u$ adds load, formally $p_s(u) = \min P_u$ and $p_f(u) = \max P_u$, where $P_u = \{p: |[t_u, t_u + w'(u)]\cap\ [p\omega, (p+1)\omega]| >0\}$.
From the statement of Algorithm~\ref{alg:S}, we show immediately by inductive bottom-up argument that if $u \in S(v)$, then $\omega \sum_{p = p_s(u)}^{p_f(u)}  c_{v} [p,u] = w'(u).$


\begin{lemma} \label{lem:heavy_zero}
Let $(s_v,t_v)_{v\in T(V)}$ be a schedule assignment computed by $\textsc{BuildStrategy}$.
For any vertices $u$ and $z$ such that $(u,t_u)$ and $(z,t_z)$ belong to the schedule at $v$, if either $u$ or $z$ is heavy, then $|[t_u,t_u+w'(u)]\cap[t_z,t_z+w'(z)]|=0$.
\end{lemma}
\begin{proof}
Note that procedure $\textsc{InsertVertex}$ is called for a heavy input vertex with its last parameter (insertion time) being a multiple of $\omega$, and the weight $w'(u)$ is a multiple of $\omega$ by definition.
Thus, the interval $[t_u,t_u+w'(u)]$ starts and ends at the beginning and end of a box, respectively.
Hence, Constraint~\eqref{con:fullboxsafe} gives the lemma.
\end{proof}
\MC{Here $u,z$ are on the same schedule? Else heavy vertices could overlap. But the notation $t_u,t_z$ are the starting time of $u,z$ on different schedule $s_u,s_z$. I think you mean the starting time of two vertices on the same schedule? May need to change notations.}
\DD{I wrote the sequence to which you refer as the lemma above and clarified --- please check, and is it ok now?}
As a consequence of Lemma~\ref{lem:heavy_zero}, if these two jobs $(u,t_u)$ and $(z,t_z)$ overlap, where $u$ and $z$ belong to the sequence assignment $S(v)$, then both of the vertices $u$ and $z$ must be light, thus:
$$
t_u > t_z -  w'(u) - \omega =  t_z + w'(z) - w'(z) -  w'(u) - \omega  \geq  t_z + w'(z) - (2c+1)\omega.
$$
We now define the measure of progress $M(x,i)$ of strategy $\cA_S$ when searching for target $x$ after $i$ queries as follows. Let $Q_i $ be the set of the first $i$ queried vertices.
Let $v_i$ be the current root of the tree, $v_i=\root{\queriedSubtree{T}{Q_i}{x}}$. Let $S_i(v) \subseteq S(v)$ be the subsequence (suffix) of $S(v)$ consisting of those vertices which have not yet been queried. Now, we define:
$$
M(x,i) =
\begin{cases}
\min_{u \in S_i(v_i)} p_s(u), & \textup{if }S_i(v_i)\neq\emptyset, \\
L,                            & \textup{if }S_i(v_i)=\emptyset.
\end{cases}
$$
We have by definition, $M(x,i) \in \{0, 1, \ldots, L-1, L\}$. We obtain the next Lemma from a following straightforward analysis of the measure of progress: every time following sequence $S(v)$ we successively complete queries with an `up' result with a total duration of at least $a$ boxes, since the queried vertices are ordered in the first place according to minimum query time, and in the second place according to query duration, the value of the minimum $p_s(u)$, for $u \in S(v)$ remaining to be queried, advances by at least $a$ boxes.
\begin{lemma} \label{lem:progress}
The measure of progress $M(x,i)$ has the following properties:
\begin{enumerate}
\item If the $(i+1)$-st query returns an `up' result, then $M(x, i+1) \geq M(x,i)$.
\item If the $(i+1)$-st query returns a `down' result, then $M(x, i+1) \geq M(x,i) - (2c+1) \omega$.
\DD{I have some problems with this. First, measure of progress is expressed in boxes, so shouldn't we have here just $2c+1$? I also think we should have a proof for this one, but I'm not sure what is the right argument; I thought that we lose one box for a down query?}
\item Suppose that between some two steps of the strategy, $i_2 > i_1$, each of the queries $(q_{i_1+1}, \ldots, q_{i_2})$ returns an `up' result, and moreover, the total cost of queries performed was at least $a\omega$, for some $a\in \N$:
$$
\sum_{j=i_1 + 1}^{i_2} w'(q_j) \geq a\omega,
$$
where $q_j = \Q_{\cA_S,j}(T,x)$. Then, $M(x, i_2) \geq M(x, i_1) +a$.
\end{enumerate}
\qed
\end{lemma}
\DD{I would suggest to remove 1. in the lemma above as it is a special case of 3.}
Since the value of $M(x,i)$ is bounded from above by $L$, we obtain from Lemma~\ref{lem:progress} that the strategy $\cA_S$ necessarily terminates when looking for target $x$ with cost at most $L\omega + (2c+1) \omega d_x$,
\[\cost_{\cA_S}(T',x) \leq  L\omega + (2c+1) \omega d_x.\]
Thus, due to the definition of $\costo$ in \eqref{eq:costo-def} and the monotonicity of of the cost of a strategy with respect to vertex weights, we obtain the following:
\begin{corollary} \label{cor:progress-concl}
For the sequence assignment computed by Algorithm~\ref{alg:S} it holds
\[\costo_{\cA_S}(T) \leq \costo_{\cA_S}(T') \leq \omega L.\]
\qed
\end{corollary}

To prove Proposition~\ref{pro:dp}, it remains to show only the stability of the sequence assignment $S$.

\begin{lemma} \label{lem:stable}
The query sequence assignment $S$ obtained by Algorithm~\ref{alg:S} is stable.
\end{lemma}
\begin{proof}

We perform the proof by induction. Following the definition of stability, assume that $v$ is the root of the remaining subtree \MC{For induction, I think we should not assume $v$ is the root. And we don't need $v$ to be root here, right?} \DD{I think it is OK -- we just consider an arbitrary subtree, skipping actually the (trivial) base case when $v$ is a leaf.} at some moment of executing $\cA_S$ on $T'$, and let $u$ be a vertex such that $u$ is a child of $v$ lying on the path from $v$ to the target $x$. We will show that following $S(u)$ always results in a subsequence of the sequence of queries performed by following $S(v)$.

Let $S^+(v)$ be the subsequence of vertices of $S(v)$ which lie in $T_u$, and let $S^-(v)$ be the subsequence of all remaining vertices of $S(v)$.
Note that $x$ belongs to $T_u$ and hence any query to a node in $S^-(v)$ gives an `up' reply.
\MC{I agree but what is "since the root of the tree may only change to a subtree of $T_u$"?}
\DD{I changed this sequence, without changing the point, I think -- is it ok now?}

We now observe the first (leftmost) difference $v'$ of the sequences $S^+(v)$ and $S(u)$. Suppose that before such a difference occurs, the common fragment of the sequences contains a query to any vertex $y$ on the path from $u$ to $x$. Then, the root of both trees moves to the same child of $y$, and the process continues identically regardless of the initial root of the tree.
Thus, such a vertex $y$ cannot occur prior the difference in sequences $S^+(v)$ and $S(u)$.

Next, suppose that the first difference between the two sequences consists in the appearance of vertex $v$ in sequence  $S^+(v)$, i.e., $v'=v$. Then, the root of the tree moves from $v$ to $u$, and the two processes proceed identically as required.
This also implies that $t_v>t_u$.

Finally, we observe that no other first difference between the sequences $S^+(v)$ and $S(u)$ is possible by the formulation of Algorithm~\ref{alg:S}.
In particular, if a triple $(\lfloor t_z\rfloor_{\omega},\lceil t_z+w'(z)\rceil_{\omega},z)$ is added to $C(u)$ in line~\ref{ln:Cv}, then the condition in line~\ref{ln:Cvcond} and $t_v>t_u$ imply that the triple $(\lfloor t_z\rfloor_{\omega},\lceil t_z+w'(z)\rceil_{\omega},z)$ is added also to the set $C(v)$.
Similarly, an insertion of a triple $(\lfloor t_z\rfloor_{\omega},\lceil t_z+w'(z)\rceil_{\omega},z)$ for $z\in V(T_u)$ into $C(v)$ implies that this triple also belongs to $C(u)$.
Due to the sorting performed in line~\ref{ln:Cvsort} of Algorithm~\ref{alg:S}, $S^+(v)=S(u)$.
\AK{some more details would perhaps be welcome --- but this is not a priority}
\DD{I rephrased a bit using only arguments from Algorithm~\ref{alg:S}; Mengchuan: I did not understand all your arguments; I could say more if you would write this part formally, if you think the above is not correct.}
\MC{Maybe:1. Let $(s_v,t_v)$ and $(s_u,t_u)$ be the  schedules from the schedule assignment obtained by Algorithm \ref{alg:S}, define $s^+_v$ be the schedule by extracting $c_v[p,x], x\in S^-(v)$ from $s_v$. According to the algorithm, there is no first difference other than query to $v$ is possible between $s^+_v$ and $s_u$. 2. Thus when translating from $(s_v,t_v), v\in V$ to $S(v), v\in V$, $s^+_v$ will be translated into a sequence that all vertices other than $v$ have the same order than the sequence translated from $s_u$, and these two are $S^+(v)$ and $S(u)$. }


The eventual deterministic coupling, which is obtained in all cases for the strategies starting at $v$ and $u$, extends by induction to the execution of $\cA_S$ for trees rooted at a vertex $v$ and its arbitrary descendant $u'$ lying on the path from $v$ to $x$, hence the claim holds.
\end{proof}



For the chosen value $\omega$, we can apply Lemma~\ref{lem:lb} with $\omega^* = \omega - \frac{1}{cn}$, obtaining:
$$
\left(\omega - \frac{1}{cn}\right) L = \omega^* L < \left(1+\frac{11}{c}\right) \opt T,
$$
thus, by Corollary~\ref{cor:progress-concl},
$$
\costo_{\cA_S}(T) \leq \left(1+\frac{11}{c}\right) \opt T + \frac{L}{cn} \leq \left(1+\frac{12}{c}\right) \opt T,
$$
where we took into account that trivially $L\leq n$ and $\opt T \geq 1$.
We thus, by Lemmas~\ref{lem:finalalg-correct}, and~\ref{lem:stable} obtain the claim of Proposition~\ref{pro:dp}.
}
\InConference{\enlargethispage{0.9cm}}

\section{\InJournal{Proof of Proposition~\ref{pro:cost-prime}: }Reducing the Number of Down-Queries}\label{sec:blackbox}
We start with defining a function $\labell\colon V\to \{1,\ldots,\lceil\log_2 n\rceil\}$ which in the following will be called a \emph{labeling of $T$} and the value $\labell(v)$ is called the \emph{label of $v$}.
We say that a subset of nodes $H\subseteq V$ is an \emph{extended heavy part} in $T$ if $H=\{v\}\cup H'$, where all nodes in $H'$ are heavy, no node in $H'$ has a heavy neighbor in $T$ that does not belong to $H'$ and $v$ is the parent of some node in $H'$.
Let $H_1,\ldots,H_l$ be all extended heavy parts in $T$.
Obtain a tree $T_C=(V_C,E_C)$ by contracting, in $T$, the subgraph $H_i$ into a node denoted by $h_i$ for each $i\in\{1,\ldots,l\}$.
In the tree $T_C$, we want to find its labeling $\labell'\colon V_C\to \{1,\ldots,\lceil\log_2 |V_C|\rceil\}$ that satisfies the following condition: for each two nodes $u$ and $v$ in $V_C$ with $\labell'(u)=\labell'(v)$, the path between $u$ and $v$ has a node $z$ satisfying $\labell'(z)<\labell'(u)$.
One can obtain such a labeling by a following procedure that takes a subtree $T_C'$ of $T_C$ and an integer $i$ as an input.
Find a central node $v$ in $T_C'$, set $\labell'(v)=i$ and call the procedure for each subtree $T_C''$ of $T_C'-v$ with input $T_C''$ and $i+1$.
The procedure is initially called for input $T$ and $i=1$.
We also remark that, alternatively, such a labeling can be obtained via vertex rankings \cite{IyerRV88,Schaffer89}.

Once the labeling $\labell'$ of $T_C$ is constructed, we extend it to a labeling $\labell$ of $T$ in such a way that for each node $v$ of $T$ we set $\labell(v)=\labell'(v)$ if $v\notin H_1\cup\cdots\cup H_l$ and $\labell(v)=\labell'(h_i)$ if $v\in H_i$, $i\in\{1,\ldots,l\}$.

Having the labeling $\labell$ of $T$, we are ready to define a query sequence $R(v)$ for each node $v\in V$.
The $R(v)$ contains all nodes $u$ from $T_v$ such that $\labell(u)<\labell(v)$ and each internal node $z$ of the path connecting $v$ and $u$ in $T$ satisfies $\labell(z)>\labell(u)$.
Additionally, the nodes in $R(v)$ are ordered by increasing values of their labels.
See Figure~\ref{fig:labeling} for an example.\InConference{\enlargethispage{1.4cm}}
\begin{figure}[!!!b]
\begin{center}
\includegraphics[scale=0.8]{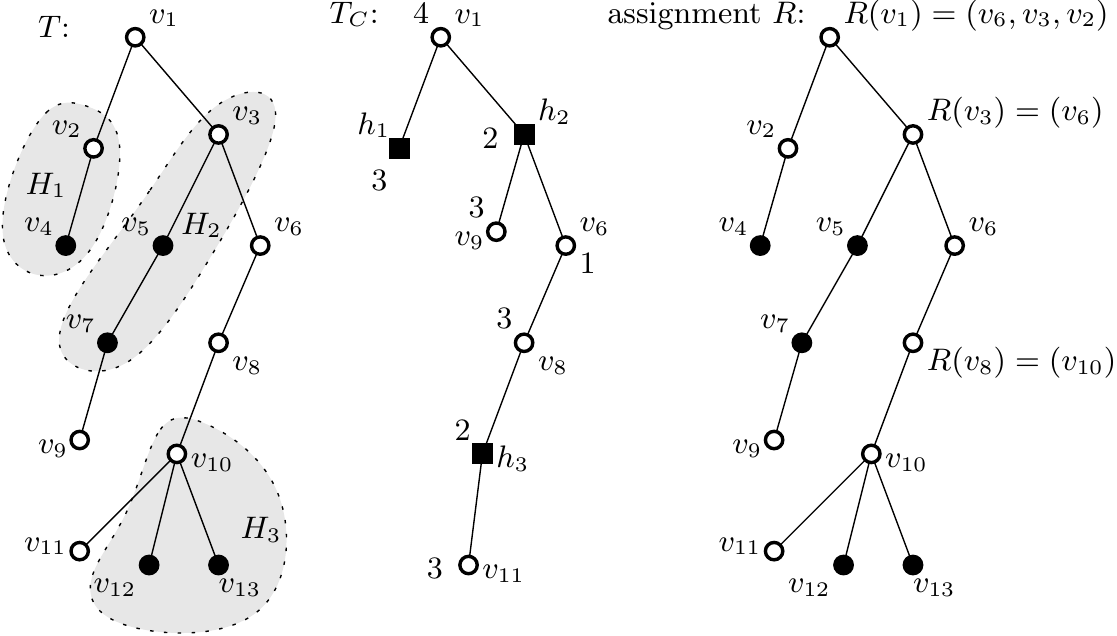}
\caption{A tree $T$ (on the left) has light vertices (marked as white nodes) and heavy ones (dark circles); also heavy extended parts are marked. The tree $T_C$ (in the middle) is used together with its labeling (integers are the labels) to obtain the sequence assignment $R$ (on the right); here we skip the sequence assignment for each node $v$ for which $R(v)=\emptyset$.}
\label{fig:labeling}
\end{center}
\end{figure}

\InJournal{
We start by making some simple observations regarding the sequence assignment $R$.
\begin{observation} \label{obs:light-in-R}
For each $v\in V$ and for each $u\in R(v)$, $\w(u)\leq c\omega$.
\eop
\end{observation}
\begin{observation} \label{obs:visibility}
For each $v\in V$, any two nodes in $R(v)$ have different labels.
\eop
\end{observation}
\begin{observation} \label{obs:R-linear}
The sequence assignment $R$ can be computed in time $O(n\log n)$.
\eop
\end{observation}
}

By $x$ we refer to the target node in $T$.
Fix $S$ to be a stable sequence assignment in the remaining part of this section and by $R$ we refer to the sequence assignment constructed above.
Then, we fix $S^+$ to be $S^+(v)=R(v)\circ S(v)$ for each $v\in V$.
\InJournal{Denote by $U_i$ the first $i$ nodes queried by $\cA_{S^+}$ and let $C_i=\min \labell(\queriedSubtree{T}{U_{i-1}}{x})$ for each $i\geq 1$.
For brevity we denote $U_0=\emptyset$ and $C_0=0$; we also denote by $u_i$ the node in $U_i\setminus U_{i-1}$, $i\geq 1$.}
A query made by $\cA_{S^+}$ to a node that belongs to $R(v)$ for some $v\in V$ is called an \emph{$R$-query}; otherwise it is an \emph{$S$-query}.
\InJournal{
\begin{lemma} \label{lem:min-label-conn}
For each $i\geq 0$, the nodes in $\queriedSubtree{T}{U_i}{x}$ with minimum label induce a connected subtree.
\eop
\end{lemma}

The next two lemmas will be used to conclude that the number of light queries performed by $\cA_{S^+}$ is bounded by $2\log_2 n$ (see Lemma~\ref{lem:R-logn}).
\begin{lemma} \label{lem:R-up}
If the $i$-th query of $\cA_{S^+}$ is an $R$-query resulting in an `up' reply, then $C_{i+1}\geq C_{i}+1$.
\end{lemma}
\begin{proof}
By construction, $u_i$ has the minimum label among all nodes in $\queriedSubtree{T}{U_{i-1}}{x}$.
By Lemma~\ref{lem:min-label-conn}, either $u_i$ is the unique node with label $\labell(u_i)$ in the tree $\queriedSubtree{T}{U_{i-1}}{x}$ or there are more nodes with this label and they all belong to a single extended heavy part in $\queriedSubtree{T}{U_{i-1}}{x}$ with $u_i$ being closest to the root of $\queriedSubtree{T}{U_{i-1}}{x}$.
In both cases, since the reply is `up', we obtain that $\queriedSubtree{T}{U_{i}}{x}$ has no node with label $\labell(u)$, which proves the lemma.
\end{proof}

\begin{lemma} \label{lem:R-down}
If the $i$-th query of $\cA_{S^+}$ is an $R$-query that results in a `down' reply, then one of the two cases holds:
\begin{enumerate}[label={\normalfont{(\roman*)}},ref={(\roman*)}]
 \item\label{it:heavy-child1} if $u_i$ has no heavy child that belongs to $\queriedSubtree{T}{U_{i}}{x}$, then $C_{i+1}\geq C_{i}+1$,
 \item\label{it:heavy-child2} if $u_i$ has a heavy child that belongs to $\queriedSubtree{T}{U_{i}}{x}$, then $C_{i+1}=C_{i}=\labell(u_i)$ and for each $j>i$ such that $C_i=C_{j+1}$, all queries $i+1,\ldots,j$ are $S$-queries.
\end{enumerate}
\end{lemma}
\begin{proof}
By Lemma~\ref{lem:min-label-conn}, the nodes with label $C_i$ induce a connected subtree in $\queriedSubtree{T}{U_{i-1}}{x}$.
This immediately implies \ref{it:heavy-child1}.
By construction, if $u_i$ has a heavy child $u'$ that is in $\queriedSubtree{T}{U_{i}}{x}$, then $u'=\root{\queriedSubtree{T}{U_{i}}{x}}$ and the labels of $u_i$ and $u'$ are the same.
The latter is due to the fact that both $u_i$ and $u'$ belong to the same extended heavy part in $T$.
Suppose for a contradiction that the $j$-th query (performed say on a node $z$) is an $R$-query and $C_{j+1}=C_i$, $j>i$.
This in particular implies that $z\in R(\root{\queriedSubtree{T}{U_{j-1}}{x}})$.
Due to Lemma~\ref{lem:R-up}, the reply to this query is `down'.
By \ref{it:heavy-child1}, $z$ has a heavy child that belongs to $\queriedSubtree{T}{U_j}{x}$.
By Observation~\ref{obs:light-in-R}, $z$ is a light node and therefore $z$ along with some of its descendants and $u_i$ with some of its descendants form two different extended heavy parts in $T$.
Since $z$ and $u_i$ have the same label, there exists a light node $u_i'$ in $T$ on the path between $u_i$ and $z$ with label smaller than $\labell(u_i)$.
Assume without loss of generality that no other node of this path that lies between $u_i$ and $u_i'$ has label smaller than $\labell(u_i')$.
The above-mentioned path is contained in $\queriedSubtree{T}{U_{i-1}}{x}$ since both $u_i$ and $z$ belong to this subtree.
This however implies that $u_i'\in R(v)$ because $\labell(u_i')<\labell(u_i)$ and no node on the path between $u_i$ and $u_i')$ has label smaller than $\labell(u_i')$.
Moreover, $u_i'$ precedes $u_i$ in $R(v)$ meaning that among one for the first $i$ queries, $u_i'$ must have been queried --- a contradiction with the fact that $u_i'$ belongs to $\queriedSubtree{T}{U_i}{x}$.
\end{proof}
\begin{lemma} \label{lem:R-logn}
For each target node, the total number of $R$-queries made by $\cA_{S^+}$ is at most $2\log_2 n$.
\end{lemma}
\begin{proof}
It follows from Lemmas~\ref{lem:R-up} and~\ref{lem:R-down} that after any two subsequent $R$-queries the value of parameter $C_i$ increases by at least $1$.
\end{proof}

The next two lemmas will be used to bound the number of $S$-queries in $S^+$ receiving a `down' reply to be at most $2\log_2 n$.
\begin{lemma} \label{lem:last-in-R}
If all nodes in $R(v)$ have been queried by $\cA_{S^+}$ after an $i$-th query for some $v\in V$ and $v$ is the root of $\queriedSubtree{T}{U_{i}}{x}$, then $\labell(v)=C_{i+1}$.
\end{lemma}
\begin{proof}
Suppose for a contradiction that $\labell(v)\neq C_{i+1}$.
Since $v$ belongs to $\queriedSubtree{T}{U_i}{x}$, we have that $\labell(v)>C_{i+1}$.
Thus, by construction, there exists a light node $u$ in $\queriedSubtree{T}{U_i}{x}$ with $\labell(u)=C_{i+1}$ such that all internal nodes on the path between $v$ and $u$ have labels larger than $\labell(u)$.
Therefore, $u$ belongs to $R(v)$ because $v$ is the root of $\queriedSubtree{T}{U_i}{x}$.
This implies that $u$ has been already queried --- a contradiction with $u$ being in $\queriedSubtree{T}{U_i}{x}$.
\end{proof}

\begin{lemma} \label{lem:S-down}
If the $i$-th query of $\cA_{S^+}$ is an $S$-query performed on a light node and the reply is `down', then $\queriedSubtree{T}{U_i}{x}$ has no light node with label $C_i$.
\end{lemma}
\begin{proof}
Suppose that the $i$-th query is performed on a node $u$ in $S(v)$ for some $v\in V$.
Clearly, $v$ is the root of $\queriedSubtree{T}{U_{i-1}}{x}$.
Since the considered query is an $S$-query, all vertices in $R(v)$ have been already queried.
Thus, by Lemma~\ref{lem:last-in-R}, $\labell(v)=C_{i}$.
By construction, $v$ is the only light node in this subtree having label $C_i$.
Since the reply to the $i$-th query is `down', $v$ does not belong to $\queriedSubtree{T}{U_i}{x}$.
\end{proof}

We are now ready to prove Proposition~\ref{prop:cost-prime}.

By Lemma~\ref{lem:R-logn},}
\InConference{In the Appendix we show that,}
 in $\cA_{S^+}$, the total number of $R$-queries does not exceed $2\log_2 n$.
\InJournal{Note that since}\InConference{We also show that} $S$ is stable, \InConference{and so} for each target node $x$, the $S$-queries performed by $\cA_{S^+}$ are a subsequence of the queries performed by $\cA_{S}$.
Therefore, the potentially additional queries made by $\cA_{S^+}$ with respect to $\cA_{S}$ are $R$-queries.
\InJournal{By Observation~\ref{obs:light-in-R}, each $R$-query is made on a light node.
By definition of function $\costo$ and Observation~\ref{obs:light-in-R}, any $R$-query increases the value of $\costo$ of $\cA_{S^+}$ with respect to the value of $\costo$ of $\cA_{S}$ by at most $\comega$.}
\InConference{We then formally show that each $R$-query is made on a light node and that any $R$-query increases the value of $\costo$ of $\cA_{S^+}$ with respect to the value of $\costo$ of $\cA_{S}$ by at most $\comega$.} Hence we have:
\mathx
\costo_{\cA_{S^+}}(T) \leq \costo_{\cA_{S}}(T) + 2\comega \log_2 n.
\mathx

\InJournal{By Lemmas~\ref{lem:R-down} and~\ref{lem:S-down},}
\InConference{Moreover, we show in the Appendix that}
the total number of queries in strategy $\cA_{S^+}$ to light nodes receiving `down' replies can be likewise bounded by $2\log_2 n$.
Since each such query introduces a rounding difference of at most $\comega$ when comparing cost functions $\cost$ and $\costo$, we thus obtain:
\mathx
\cost_{\cA_{S^+}}(T) \leq \costo_{\cA_{S^+}}(T) + 2\comega \log_2 n.
\mathx

Combining the above observations gives the claim of the Proposition.

\InJournal{
\section{Proof of Theorem~\ref{thm:recursive-algo}: A \texorpdfstring{$O(\sqrt{\log n})$}{O(sqrt(log n))}-Approximation Algorithm}
\label{sec:algosqrt}
We start with some notation.
Given a tree $T=(V,E,w)$ and a fixed value of parameter $\alpha$, we find a subtree $T^*=(V^*,E^*)$ of the input tree $T$, called an \emph{$\alpha$-separating tree}, that satisfies: $\root{T^*}=\root{T}$ and each connected component of $T\setminus V^*$ has at most $\alpha$ vertices.
An $\alpha$-separating tree $T^*$ is \emph{minimal} if the removal of any leaf from $T^*$ gives an induced tree that is not an $\alpha$-separating tree.
Then, for a target node $x\in V$, we introduce a recursive strategy $\cR$ that takes the following steps:
\begin{enumerate}
\item $\cR$ first applies strategy $\cA^*$ restricted to tree $T^*$ to locate the node $x'$ of $T^*$ which is closest to the target $x$.
\item Then, $\cR$ queries $x'$, which either completes the search in case when $x'$ is the target or provides a neighbor $x''$ of $x'$ that is closer to the target than $x'$.
\item If $x'$ is not the target, then the strategy calls itself recursively on the subtree $T_{x''}$ of $T\setminus\{x'\}$ containing $x$. The latter strategy for $T_{x''}$ is denoted by $\cR_{x''}$.
(Note that $T_{x''}$ is a connected component in $T\setminus V^*$.)
\end{enumerate}
Such a search strategy $\cR$ obtained from $\cA^*$ and strategies $\cR_{\root{T'}}$ (constructed recursively) for subtrees $T'$ in $T\setminus V^*$ is called a \emph{$(\cA^*,\{\cR_{\root{T'}}\stX T' \in \mathcal{C}(T \setminus V^*)\})$-strategy}, where $\mathcal{C}(T \setminus V^*)$ is the set of connected components (subtrees) in $T \setminus V^*$.

The following bound on the cost of the strategy $\cR$ follows directly from the construction:
\begin{lemma} \label{lem:recursion-cost}
For a $(\cA^*,\{\cR_{\root{T'}}\stX T' \in \mathcal{C}(T \setminus V^*)\})$-strategy $\cR$ for $T$ it holds
$$
\cost_\cR (T) \leq \cost_{\cA^*} (T^*) + \max_{x' \in V^*} w(x') + \max_{T' \in \mathcal{C}(T \setminus V^*)} \cost_{\cR_{\root{T'}}} (T').
$$
\qed
\end{lemma}

We now formally describe and analyze the aforementioned contractions of subpaths in a tree.
A maximal path with more than one node in a tree $T$ that consists only of vertices that have degree two in $T$ is called a \emph{long chain in $T$}.
For each long chain $P$, contract it into a single node $v_P$ with weight $\min_{u\in V(P)}w(u)$, obtaining a tree $\longchains{T}$.
In what follows, the tree $\longchains{T}$ is called a \emph{chain-contraction of} $T$.

Our first step is a remark that, at the cost of losing a multiplicative constant in the final approximation ratio, we may restrict ourselves to trees that have no long chains.
This is due to the following observation.
\begin{lemma} \label{lem:no-long-chains}
Let $T$ be a tree.
Given a $p$-approximate search strategy for $\longchains{T}$, a $(p+1)$-approximate search strategy for $T$ can be computed in polynomial time.
\end{lemma}
\begin{proof}
Let $\cA'$ be a search strategy for $\longchains{T}$.
We obtain a search strategy $\cA$ for $T$ in two stages.
In the first stage we `mimic' the behavior of $\cA'$: (i) if $\cA'$ queries a node $v$ that also belongs to $T$, then $\cA$ also queries $v$; (ii) if $\cA'$ queries a node $v_P$ that corresponds to some long chain $P$ in $T$, then $\cA$ queries, in $T$, a node with minimum weight in $P$.
Note that after the first stage, the search strategy either located the target or determined that the target belongs to a subpath $P'$ of some long chain $P$ of $T$.
Moreover, the total cost of all queries performed in the first stage is at most $\cost_{\cA'}(\longchains{T})$.

Then, in the second stage we compute (in $O(n^2)$-time) an optimal search strategy $\cA_{P'}$ for $P'$ \cite{CicaleseJLV12}.
Due to the monotonicity of the cost over taking subgraphs, $\cost_{\cA_{P'}}(P')=\opt{P'}\leq\opt{T}$.

Both stages provide us with a search strategy for $T$ with cost at most $\cost_{\cA'}(\longchains{T})+\opt{T}$.
Since, $\opt{\longchains{T}}\leq\opt{T}$ and $\cost_{\cA'}(\longchains{T})\leq p\cdot\opt{\longchains{T}}$, the lemma follows.
\end{proof}
Note that it is straightforward to verify whether any vertex $v$ of $T$ is a leaf in the $\alpha$-separating tree of $T$ and hence we obtain the following.
\begin{observation} \label{obs:seperating-set}
Given a tree $T$ with no long chain and $\alpha$, a minimal $\alpha$-separating tree of $T$ can be computed in polynomial-time.
\qed
\end{observation}

Using Lemma~\ref{lem:no-long-chains} and choosing appropriately the value of $\alpha$, one can obtain an $\alpha$-separating tree of $T$ having at most $t = 2^{O(\sqrt{\log n})}$ vertices.

\begin{lemma} \label{lem:minimal-separating}
Let $T$ be any tree and let $\alpha$ be selected arbitrarily.
If $T^*$ is a minimal $\alpha$-separating tree of $T$, then $\longchains{T^*}$ has at most $4\left\lceil\frac{n}{\alpha}\right\rceil$ vertices.
\end{lemma}
\begin{proof}
By definition, for each leaf $v$ of $T^*$, the subtree $T_v$ has more than $\alpha$ nodes.
Since these trees are node-disjoint, we obtain that there are at most $\lceil \frac{n}{\alpha} \rceil$ leaves in $T^*$.
We denote the leaves of $T^*$ by $v_1,v_2,\ldots,v_{l}$, $l\leq \left\lceil \frac{n}{\alpha} \right\rceil$; note that $\longchains{T^*}$ has the same leaves as $T^*$.

Let $V(\longchains{T^*})$ be the vertex set of $\longchains{T^*}$.
Then, we claim that $\card{V(\longchains{T^*})}=O(\lceil \frac{n}{\alpha} \rceil)$ by counting the number of nodes with different degrees in $\longchains{T^*}$.
Clearly, we have $\card{\{v\in V(\longchains{T^*})\stX \deg(v)>2\}} \leq \lceil \frac{n}{\alpha} \rceil$.
Since the tree $\longchains{T^*}$ contains no long chains, the parent (if exists) of every node with degree exactly $2$ must have degree at least $3$.
Thus,
$$\card{\{v\in V(\longchains{T^*})\stX \deg(v)=2\}} \leq \card{\{v\in V(\longchains{T^*})\stX \deg(v)>2\}}+1 \leq \left\lceil \frac{n}{\alpha} \right\rceil+1.$$
Hence we get $\card{V(\longchains{T^*})}\leq 4 \left\lceil \frac{n}{\alpha} \right\rceil$.
\end{proof}

With Lemmas~\ref{lem:no-long-chains},~\ref{lem:minimal-separating} and Observation~\ref{obs:seperating-set} we are now ready to obtain the efficient recursive decomposition of the problem:
\begin{lemma} \label{lem:recursive-algo}
If there is a $O(1)$-approximation algorithm running in $n^{O(\log n)}$ time for any input tree, then one can obtain a $O(\sqrt{\log n})$-approximation algorithm with polynomial running time for any input tree.
\end{lemma}

\begin{proof}
Suppose $\textsc{Solve}$ is a given constant-factor approximation algorithm running in time $n^{O(\log n)}$ that, for any input tree $T$, outputs a search strategy for $T$.
We then design a polynomial-time procedure $\textsc{Rec}$ as shown in Algorithm~\ref{alg:rec}, which outputs a search strategy $\cR$ for an input tree $T$.

\begin{algorithm}
\caption{$O(\sqrt{\log n})$)-approximation procedure $\textsc{Rec}$ based on $n^{O(\log n)}$-time constant approximation algorithm $\textsc{Solve}$}\label{alg:rec}
\begin{algorithmic}[1]

\Procedure{Rec}{tree $T=(V,E,w)$}
\State $n\gets \card{V}$
\If {$n\leq 2^{\sqrt{\log n}}$}
	\State \textbf{return} \Call{Solve}{$T$} \label{line:solve}
\Else
\State $\alpha \gets n/2^{\sqrt{\log n}}$
\State $T^* \gets$ a minimal $\alpha$-separating tree of $T$ with vertex set $V^*$
\State $\cA^* \gets$ \Call{Solve}{$\longchains{T^*}$}
\label{line:solve2}
\State $\cA_{T^*} \gets$ search strategy for $T^*$ obtained from $\cA^*$ as described in proof of Lemma~\ref{lem:no-long-chains}
\For{\textbf{each} $T'$ in $\mathcal{C}(T \setminus V^*)$}
\State $\cR_{\root{T'}} \gets$ \Call{Rec}{$T'$};
\EndFor
\State \textbf{return} $(\cA_{T^*},\{\cR_{\root{T'}}\stX T' \in \mathcal{C}(T \setminus V^*)\})$-strategy for $T$
\EndIf
\EndProcedure
\end{algorithmic}
\end{algorithm}

Each call to $\textsc{Solve}$ in line~\ref{line:solve} has running time $(2^{\sqrt{\log n}})^{O(\log (2^{\sqrt{\log n}}))}$, which is a polynomial in $n$.
The same holds for each call call in line~\ref{line:solve2} because, by Lemma~\ref{lem:minimal-separating}, $\longchains{T^*}$ has at most $4 \left\lceil \frac{n}{\alpha} \right\rceil = O(2^{\sqrt{\log n}})$ vertices.
Thus, procedure $\textsc{Rec}$ has polynomial running time and it remains to bound the cost of the search strategy $\cR$ computed by $\textsc{Rec}$.

To bound the recursion depth of $\textsc{Rec}$, note that each time a recursive call is made, the size of instance (input tree) decreases $2^{\sqrt{\log n}}$ times.
Thus, the depth is bounded by $\log_{(2^{\sqrt{\log n}})}n=\sqrt{\log n}$.
In the search strategy computed by procedure $\textsc{Rec}$, at each level of the recursion we execute the search strategy computed by one call to $\textsc{Solve}$ and one vertex of the $(n/2^{\sqrt{n}})$-separating tree is queried.
This follows from the definition of $(\cA_{T^*},\{\cR_{\root{T'}}\stX T' \in \mathcal{C}(T \setminus V^*)\})$-strategy.
By Lemma~\ref{lem:no-long-chains},
\[\cost_{\cA_{T^*}}(T^*)\leq c'\cdot\opt{T^*}\]
for some constant $c'$.
By Lemma~\ref{lem:recursion-cost} and since $\opt{T^*}\leq\opt{T}$, the cost of $\cR$ at each recursion level is bounded by $(c'+1)\opt{T}$.
This gives that $\cost_{\cR}(T)\leq c'\sqrt{\log n}\cdot\opt{T}$ as required.
%
%
\end{proof}

Noting that the existence of a constant-approximation procedure with $n^{O(\log n)}$ running time follows from Theorem~\ref{thm:qptas} (by taking $\varepsilon=1$), the claim of Theorem~\ref{thm:recursive-algo} follows directly from Lemma~\ref{lem:recursive-algo}.
}

\InJournal{
\section*{Acknowledgment}

The authors thank Jakub \L{}ącki for preliminary discussions on the studied problem.
}

\InConference{\newpage}

\bibliographystyle{plain}
\bibliography{queries}

\end{document}